\documentclass{sn-jnl}
\usepackage{graphicx}%
\usepackage{multirow}%
\usepackage{amsmath,amssymb,amsfonts}%
\usepackage{amsthm}%
\usepackage{mathrsfs}%
\usepackage[title]{appendix}%
\usepackage[dvipsnames]{xcolor}%
\usepackage{textcomp}%
\usepackage{manyfoot}%
\usepackage{booktabs}%
\usepackage{algorithm}%
\usepackage{algorithmicx}%
\usepackage{algpseudocode}%
\usepackage{listings}%
\usepackage{braket}%
\usepackage{mathtools}%
\usepackage{tikz}
\usepackage{pgf}
\usepackage{lmodern}
\usepackage{import}
\usepackage{comment}
\usepackage{booktabs}

\usetikzlibrary{quantikz2}
\usetikzlibrary{external}
\newcommand{\subgrgen}[1]{\left\langle #1 \right\rangle}
\DeclarePairedDelimiter\abs{\lvert}{\rvert}%
\DeclarePairedDelimiter\norm{\|}{\|}%

\def\diffd{\mathrm{d}} %
\DeclareDocumentCommand\differential{ o g d() }{ %
	\IfNoValueTF{#2}{
		\IfNoValueTF{#3}
			{\diffd\IfNoValueTF{#1}{}{^{#1}}}
			{\mathinner{\diffd\IfNoValueTF{#1}{}{^{#1}}\argopen(#3\argclose)}}
		}
		{\mathinner{\diffd\IfNoValueTF{#1}{}{^{#1}}#2} \IfNoValueTF{#3}{}{(#3)}}
	}
\DeclareDocumentCommand\dd{}{\differential} %

\newcommand{\inner}[2]{\left\langle {#1}, {#2} \right \rangle}

\newcommand{\rev}[1]{#1}

\makeatletter
\def\footnoterule{
  \hrule \@width 0.4 in \kern 3\p@} 
\makeatother

\newtheorem{theorem}{Theorem}%
\newtheorem{proposition}[theorem]{Proposition}%
\newtheorem{lemma}[theorem]{Lemma}%

\newtheorem{definition}[theorem]{Definition}%

\begin{document}

\title{Mitigating exponential concentration in covariant quantum kernels for subspace and real-world data}

\author*[2]{\fnm{Gabriele} \sur{Agliardi}}\email{gabriele.agliardi@it.ibm.com}
\author[1]{\fnm{Giorgio} \sur{Cortiana}}
\author[3]{\fnm{Anton} \sur{Dekusar}}
\author[1]{\fnm{Kumar} \sur{Ghosh}}
\author[1]{\fnm{Naeimeh} \sur{Mohseni}}
\author*[1]{\fnm{Corey} \sur{O'Meara}}\email{corey.o'meara@eon.com}
\author[3]{\fnm{V\'ictor} \sur{Valls}}
\author[4]{\fnm{Kavitha} \sur{Yogaraj}}
\author[3]{\fnm{Sergiy} \sur{Zhuk}}

\affil[1]{\orgdiv{E.ON Digital Technology GmbH}, \state{Hanover}, \country{Germany}}
\affil[2]{\orgdiv{IBM Quantum, IBM Research -- Italy}}
\affil[3]{\orgdiv{IBM Quantum, IBM Research Europe -- Dublin}}
\affil[4]{\orgdiv{IBM Quantum, IBM Research -- India}}

\abstract{
Fidelity quantum kernels have shown promise in classification tasks, particularly when a group structure in the data can be identified and exploited through a covariant feature map. In fact, there exist classification problems on which covariant kernels provide a provable advantage, thus establishing a separation between quantum and classical learners.
However, their practical application poses two challenges: on one side, the group structure may be unknown and approximate in real-world data, and on the other side, scaling to the `utility' regime (above $100$ qubits) is affected by exponential concentration. 
In this work, we propose a novel error mitigation strategy specifically tailored for fidelity kernels, called Bit Flip Tolerance (BFT), to \rev{alleviate} exponential concentration in \rev{our} utility-scale experiments.
\rev{We apply fidelity kernels to real-world data with unknown structure, related to the scheduling of a fleet of electric vehicles, and to synthetic data generated from the union of subspaces, which is then close to many relevant real-world datasets.}
Our multiclass classification reaches accuracies comparable to classical SVCs up to $156$ qubits, thus constituting the largest experimental demonstration of quantum machine learning on IBM devices to date. %
For the real-world data experiments, \rev{the proposed BFT proves useful} 
on $40+$ qubits, where mitigated accuracies reach $80\%$, in line with classical, compared to $33\%$ without BFT. Through the union-of-subspace synthetic dataset with $156$ qubits, we demonstrate a mitigated accuracy of $80\%$, compared to $83\%$ of classical models, and $37\%$ of unmitigated quantum, using a test set of limited size.

}

\newgeometry{textwidth=144 mm}
\maketitle

\newgeometry{twocolumn, textwidth=174 mm}

\clearpage
\section{Introduction}\label{sec:intro}
Kernel methods are a powerful technique with applications in Support Vector Classifiers (SVCs)~\cite{boser1992training}.   
In brief, kernels empower SVCs to classify non-linearly separable data by mapping it into a higher-dimensional space~\cite{scholkopf2018learning}, where data is linearly separable. The performance of a specific kernel is closely connected to the structure of the data. For example, polynomial kernels often perform best in image processing and speech recognition tasks \cite{meijering1999image}, while tangent kernels are better suited for neural-network-inspired SVC.

Over the last few years, quantum kernels have received increased attention because they can extend the capabilities of their classical counterparts. In short, while the task of classical and quantum kernels is the same (i.e., map data into a higher dimensional space), quantum kernels can use quantum states and transformations to map data into a quantum Hilbert space: a much larger space that is generally not reachable by classical computers~\cite{PhysRevLett.113.130503, huang2021power}. %
Fidelity kernels~\cite{havlivcek2019supervised} provide a provable speedup over classical methods, at least for artificially constructed problems~\cite{liu2021rigorous}. Despite fidelity quantum kernels are promising, %
their application to real-world problems is yet at an initial stage~\cite{slattery_numerical_2023,krunic_quantum_2022}. Two important blockers limit the effective use of fidelity quantum kernels in practice. %
First, fidelity quantum kernels are currently known to have an ideal behavior when data has a group structure and the feature maps is `covariant' with the structure~\cite{glick_covariant_2024}. In fact, such structure was demonstrated only through synthetic datasets~\cite{liu2021rigorous, glick_covariant_2024}.
Second, it is challenging to implement fidelity quantum kernels successfully at \emph{utility scale} with real-world data on the current noisy quantum hardware. Technically, quantum kernels suffer from the exponential concentration~\cite{thanasilp2024exponential}, which results in poor generalization capabilities. %
This issue cannot be resolved by merely increasing the dataset size. Instead, one must increase the number of measurement shots exponentially, which is impractical for problems that involve 100+ qubits. %

To bring fidelity quantum kernels to practical problems at utility scale, we address the two challenges above through the following contributions.

\begin{itemize}
\item \textbf{Covariant kernels.} %
We ask the question: %
Are fidelity kernels suited to exploit other types of data structures?
In Prop.~\ref{prop:group}, we show that an exact match between a fidelity kernel and the ideal kernel, is always associated with the presence of the particular `covariant' group structure in the quantum unitaries defining the feature maps. 
At the same time, though, it is not necessary to know said structure a priori. Practically, this implies that it is always possible to experiment fidelity kernels along with other classical or quantum kernels, in the quest of the best classification performance. In the case of fidelity kernels, the prescribed structure may emerge from the combination of the data with the encoding chosen for the data itself.

\item \textbf{Concentration and bit-flip tolerance (BFT).}  %
We propose a \rev{noise} mitigation strategy specifically tailored to fidelity kernels. %
In general, the value of a fidelity kernel on a pair of data points, is the fidelity between two statevectors representing the two data points. In turn, the fidelity is the probability of sampling the bitstring $0^n$ from the associated quantum circuit. The practical estimation of said probability through measurements, though, becomes difficult when the number $n$ of qubits grows. Therefore in the proposed method we introduce a bit flip tolerance (BFT), namely we count the bitstrings that are \textit{close} to $0^n$ in the Hamming distance. %
\rev{The proposed technique also alleviates concentration, since noise is one of its root causes~\cite{thanasilp2024exponential}.}

\item \textbf{Experiments.} We apply kernel alignment~\cite{glick_covariant_2024} to two different types of multi-class classification problems: (i) a real-world dataset related to how electric vehicles should charge or discharge to maximize a utility function~\cite{agliardi2024machine}, and (ii) an artificial dataset having the structure of the union of subspaces. 
The latter dataset is introduced to test the applicability of the BFT method at utility scale, since increasing the number of features beyond a given point, does not necessarily improve the classification accuracy with our real-world data. Instead of generating a synthetic covariant dataset tout court, we opt to employ a dataset from the union of subspaces, as datasets of this form are known to be practically relevant~\cite{elhamifar2013sparse, BahadoriKFL15, Khodadadzadeh7120510}, while showing good properties that connect them to the fidelity kernels.
The experiments provide the following insights: (i) Our calibration of the BFT suggests that the threshold for the Hamming weight should be scaled linearly with the number of qubits. (ii) The 10-qubit quantum kernel implemented on noisy hardware with real data performs comparably to a classical kernel, achieving an accuracy of $87\%$, compared to $80\%$ of our classical benchmark. At this scale, \rev{performance without BFT is already comparable to classical, and BFT does not degrade it}.
Experiments are limited to a single instance due to \rev{the significant resources required by hardware execution}. (iii) Comparable results between quantum and classical are also observed for synthetic data at utility scale, with an accuracy of $80\%$, against $83\%$ of the classical model, on 156 features, thus providing the largest quantum machine learning experiment on IBM hardware to-date. This suggests that fidelity kernels may be exploited for the well studied class of subspace clustering problems~\cite{vidal2011subspace, elhamifar2013sparse}, opening future avenues for research. (iv) The BFT method is essential for good accuracy on $40+$ qubits. In fact, real-world data with $40$ qubits have an accuracy up to $80\%$ with BFT, dropping to $33\%$ without BFT. Similarly, on 100 qubits, experiments reach $80\%$ to $100\%$ accuracy with BFT, compared to $37\%$ without BFT.

\end{itemize}

\section{Preliminaries}\label{sec:preliminaries}

\subsection{Kernels}

A \textit{kernel} $k(x, x')$ is a function that maps two input data points $x$ and  $x'$ from an input space $\mathcal{X}$ to a real number, quantifying their similarity. 
In Machine Learning, kernels are typically required to be positive semi-definite.  This means that for any finite sequence $\{x_i\}_{i=1}^M$, with $x_i \in \mathcal{X}$ for all $i$, the \textit{kernel matrix} $K_{ij} = k(x_i, x_j)$ is positive semi-definite. The restriction to positive semi-definite kernels simplifies empirical risk minimization through the representer theorem in statistical learning~\cite{mohri_foundations_2018}, and therefore the fitting of classification models. By Mercer's theorem~\cite{mohri_foundations_2018}, positive semi-definiteness of kernels is equivalent to the existence of a function $\phi: \mathcal{X} \to \mathcal{H}$, called \textit{feature map}, valued in a high-dimensional feature space $\mathcal{H}$, such that $k(x, x') = \langle \phi(x), \phi(x') \rangle$, where $\langle \cdot, \cdot \rangle$ denotes the inner product defined in $\mathcal{H}$.

A two-class support vector classifier (SVC) finds the hyperplane in the feature space $\mathcal H \supset \phi(\mathcal X)$ that best separates the classes~\cite{shawe-taylor_kernel_2004, mohri_foundations_2018}.
The extension of SVCs from binary classification to multi-class, typically reduces to the the application of multiple binary problems~\cite{mohri_foundations_2018}. In this paper, we use the one-vs-one (1-vs-1) approach, where a binary classification is performed for each pair of classes. Each model assigns a vote to the predicted class for a given sample, and after tallying the votes, the class with the highest score is selected as the final prediction. %

The performance of quantum kernels is highly dependent on the choice of kernel, and selecting an appropriate kernel for a classification problem can be challenging. Therefore, it is crucial to design quantum kernels that exploit the underlying structure of the dataset. Kernel alignment~\cite{shawe-taylor_kernel_2004} was introduced to address this challenge. %
It is a measure $\mathcal{A}(K_t, K)$ of the similarity between two kernel matrices, particularly useful to quantify the agreement of a given kernel function $k$ against a target or ideal kernel $k_t$, which is based on the target labels. In practice, the ideal kernel can only be computed on the training set, due to the availability of the labels. High kernel alignment indicates that $k$ closely matches the target $k_t$, suggesting that $k$ is well-suited for the task at hand. If the kernel is parametric, i.e.~$k=k_\lambda$, the set of parameters that maximizes alignment can lead to better generalization in classification tasks. Such a maximization process is again called kernel alignment.
We rely on a refined definition of kernel alignment, the \textit{centered alignment}~\cite{cortes2012algorithms}, for a reason that will be clear in Sec.~\ref{subsec:multiclass-qka}. Given two kernel matrices $K_t$ and $K$, where $K_{t,ij} = k_t(x_i, x_j)$ and $K_{ij} = k(x_i, x_j)$ for a set of data samples  $\{x_i\}_{i=1}^m$, the centered alignment $\mathcal{A}(K_t, K)$ is:

\begin{equation}\label{eq:kernel_alignment}
\mathcal{A}(K_t, K) := \frac{\langle K_t^c, K^c \rangle_F}{\|K_t^c\|_F \|K^c\|_F }
\end{equation}
where $K_t^c$ and $K^c$ are the centered versions of $K_t$ and $K$, as defined below, $\langle \cdot, \cdot \rangle_F$ denotes the Frobenius product, and $||\cdot||_F$ the Frobenius norm. Centering is performed by subtracting the mean of each row and column and adding the overall mean: $K^c := K - \mathbf{1} K / m - K \mathbf{1} / m + \mathbf{1} K \mathbf{1} / m^2$, where $\mathbf{1}$ is a matrix of ones. The centered alignment can be equivalently defined by centering one kernel matrix only, instead of both~\cite{cortes2012algorithms}.

\subsection{Fidelity quantum kernels}\label{subsec:fidelity-qk}

Quantum kernels have the potential to outperform classical kernels by leveraging the high-dimensional nature of quantum systems. In this work, we focus on fidelity quantum kernels~\cite{havlivcek2019supervised}, which quantify the similarity between quantum states to classify data. A key element of quantum kernels is the quantum feature map%
\footnote{In quantum literature~\cite{liu2021rigorous,slattery_numerical_2023}, the term quantum feature map is typically referred to $x \mapsto \ket{\Phi(x)}$. Note that such definition does not agree with that of the feature map we gave in the previous subsection, that verifies the PSD property. For fidelity kernels, the feature map to verify PSD is $\phi(x) := (\bra{0} U(x)^\dag) \otimes (U(x) \ket{0})$. Indeed, for such $\phi$, one can write $\langle \phi(x), \phi(x') \rangle = k(x,x')$, where the inner product is defined in the doubled space, and $k$ is that of Eq.~\eqref{eq:fqk}.}
that maps the classical data $x \in \mathcal{X}$ (e.g., $\mathcal{X} \subset \mathbb{R}^d$) to an $n$-qubit quantum feature state $\ket{\Phi(x)} = U(x) \ket{0^n}$ with $ U(x)$ being a parameterized circuit family. The quantum kernel, which measures the similarity between two data points, is then calculated using the fidelity between their corresponding quantum states as:
\begin{equation}\label{eq:fqk}
k(x, x') = \abs{\braket{0^n | U(x)^\dagger U(x') | 0^n} }^2.
\end{equation}
When $\mathcal{X}$ is a subgroup of a group $\mathcal G$, one can define a covariant feature map~\cite{glick_covariant_2024}, where each data point $x \in \mathcal{X}$ is mapped to a quantum unitary via a unitary representation \( D(x) \) of the group $\mathcal G$. The unitaries $D(x)$, for all $x$, are applied to a fixed state \( |\psi\rangle = V|0^n\rangle \) called fiducial state, thus producing the feature map circuit \( U(x) = D(x) V \).
If classes are distinct cosets of a subgroup $\mathcal S$ in $\mathcal G$, and the same property is preserved in terms of unitaries, then one may choose the fiducial state to be invariant under $D(\mathcal S)$, and the kernel classifies data perfectly~\cite{glick_covariant_2024}. In Subsec.~\ref{subsec:group} we better formalize this requirement of the group structure, and prove that the vice versa also holds: if a fidelity kernel equals the ideal kernel, then there exists a covariant group structure in the classes.

The effectiveness of covariant quantum kernels depends heavily on the choice of the fiducial state. If sufficient structural knowledge of the data is available, the invariant fiducial state can be selected a priori.
When working with real-world datasets, where the optimal embedding and fiducial state are not known in advance, selecting or optimizing them becomes a critical challenge. In such cases, aligning the kernel~\cite{cristianini2001kernel} with the structure of the data is essential for success. To optimize the fiducial state, a parametrized quantum circuit with parameters $\lambda$ is then implemented, namely \( |\psi_{\lambda}\rangle =  V_{\lambda}|0^n\rangle \). This results in a parameterized quantum kernel $k_{\lambda}(x, x')$ that adapts to the data structure. The kernel in Eq.~\eqref{eq:fqk} takes its final form
\begin{equation}\label{eq:fqk-align-classicaldata}
\begin{aligned}
    k_\lambda(x, x') &=& \abs{\braket{0^n | V_\lambda ^\dagger D(x)^\dagger D(x') V_\lambda | 0^n} }^2\\
    &=& \abs{ \braket{\psi_\lambda | D(x)^\dagger D(x') | \psi_\lambda} }^2.
\end{aligned}
\end{equation}

As a final remark, note that good performance of a kernel can be achieved in practice also when the kernel is not ideal, see Ref.~\cite{havlivcek2019supervised} for an example.

\section{Methods}\label{sec:methods}
\subsection{Multi-class quantum kernel alignment}\label{subsec:multiclass-qka}
\begin{figure*}[t]
\centering
\includegraphics[width=.9\linewidth]{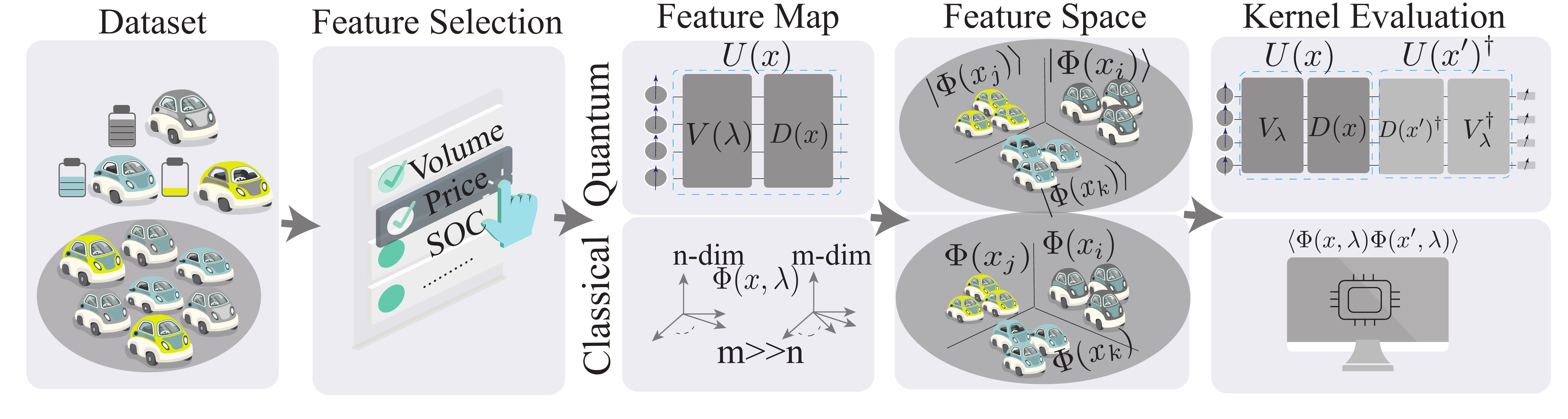} 
\caption {A schematic representation of workflow.  From the dataset related to charging electric vehicles the most important features are selected. The data is then mapped to a higher-dimensional feature space using both quantum and classical variational feature maps. Quantum and classical kernels are subsequently trained and evaluated (on quantum and classical hardware respectively) to assess their performance in classification. %
\label{fig1}}
\end{figure*}

Let us assemble the key elements introduced in the preliminaries, and describe the proposed quantum classification algorithm, also sketched in Fig.~\ref{fig1}. %
The vector $x_i$ is mapped to a quantum unitary through a family of variational circuits parametrized over the entries of $x_i$, for instance
$D(x_i) = \bigotimes_{k=1}^n R_X(x_i^{(k)}),$
where $n$ is the number of features (and qubits).
The kernel in Eq.~\eqref{eq:fqk-align-classicaldata} is then evaluated by running multiple shots of the quantum circuits. To fit the QSVC, thence, the kernel matrix is computed on the training set, equipped with corresponding labels, and then the training proceeds exactly as for a classical SVC.
To our knowledge, this is the first time QSVC is applied to multiple classes.

The kernel is aligned through an optimizer to the target kernel, meaning that $\lambda$ is chosen as to best match the ideal kernel. 
In the context of quantum kernels, the following $k_t$ is a natural choice for the target kernel, since all kernel values in Eq.~\eqref{eq:fqk-align-classicaldata} are nonnegative:
$$
k_t(x, x') =
\begin{cases}
    1 &\text{if $x, x'$ in same class},\\
    0 &\text{otherwise,}
\end{cases}
$$
Remarkably, since we use the centered version of kernel alignment in Eq.~\eqref{eq:kernel_alignment}, $k_t$ defined above is equivalent to the following other ideal kernel $k'_t$, commonly used in multi-class classical machine learning~\cite{Camargo2009}:
$$
k'_t(x, x') =
\begin{cases}
    1 &\text{if $x, x'$ in same class},\\
    -\frac{1}{C-1} &\text{otherwise,}
\end{cases}
$$
where $C$ is the number of classes. The equivalence between $k_t$ and $k'_t$ in terms of centered kernel alignment is proven in Prop.~\ref{prop:app:kernels} of the Appendix, by leveraging the invariance of centered alignment under affinities.

For the circuit design, we follow the approach from Ref.~\cite{glick_covariant_2024}. As shown in Fig.~\ref{fig:featuremap}, the feature map includes a fiducial state preparation layer and a data embedding layer. The entangler is designed to maximize the number of two-qubit gates while maintaining minimal circuit depth to reduce noise effects. It is constructed in an iterative fashion. First, we select a connected subgraph of the coupling map with the same number of qubits as the problem size. Next, we fix a starting qubit in the subgraph and we build a spanning tree by applying breadth-first search. The previous step is iterated over all qubits. Finally, we choose the tree with the smallest depth, and construct an entangler circuit by putting hardware-native two-qubit gates on the edges of the selected tree.
We map the features by assigning the most important one to the center qubit, then alternately placing the remaining features symmetrically around it. The resulting circuit has a number of trainable parameters equal to three times the number of qubits.

\begin{figure}[t]
\centering
\includegraphics[width=\linewidth]{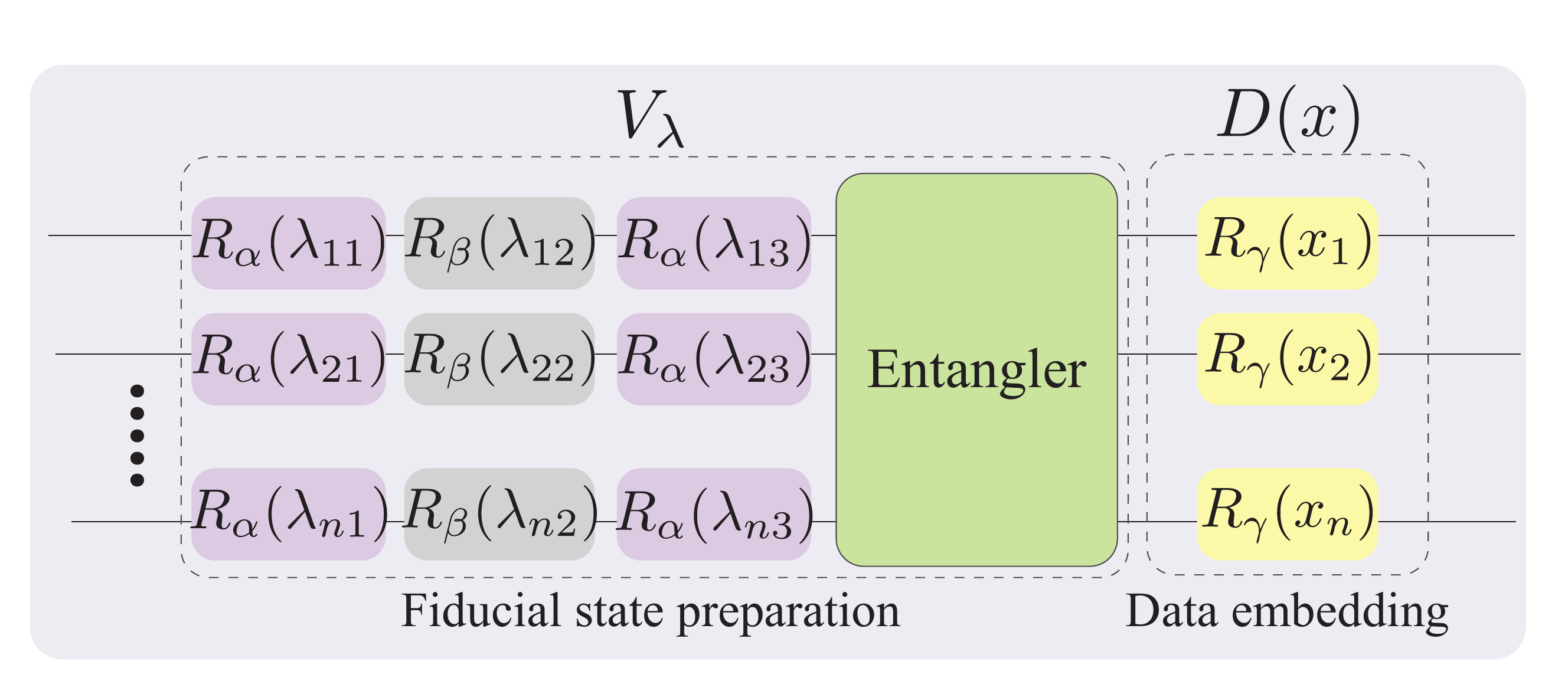} 
\caption { 
Feature map used for QSVC. The feature map $U (x)$ is defined as 
the multiplication of the unitary representation $D(x)$ applied for embedding and a circuit $V_{\lambda}$ that prepares the fiducial state. The entangler is designed  such that it maximizes the number of two-qubit gates while maintaining minimal circuit depth and taking into account the coupling map of the hardware. For experiment on \texttt{ibm\_brussels} the parameters are set as $\alpha=X, \beta=Y, \gamma=Z$. For experiment on \texttt{ibm\_fez} and \texttt{ibm\_marrakesh} the parameters are $\alpha=Z, \beta=Y, \gamma=X$. The variation in gates arises from differences in the gate alphabet implemented on the hardware.
\label{fig:featuremap}}
\end{figure}

\subsection{Noise and concentration: bit-flip tolerance}
As we scale up our experiments, implementing effective mitigation schemes is crucial. We employ dynamical decoupling~\cite{PhysRevApplied.20.064027} as implemented in Qiskit~\cite{qiskit2024} to counteract noise using an $XpXm$ pulse sequence. Additionally, we use measurement twirling~\cite{PhysRevA.105.032620} to mitigate readout biases by averaging out coherent errors during the measurement process. 

Even with these mitigation techniques, our preliminary experiments reveal a significant issue: when the number of qubits scales beyond $60$, the state $\ket{0^n}$ is never measured as an outcome of the circuit $V_\lambda^\dag D(x)^\dag D(x') V_\lambda \ket{0^n}$ in Eq.~\eqref{eq:bitflip-tol-kernel}, for $x \neq x'$, even after thousands of shots. %
The presence of residual unmitigated noise sources should be pointed out, prominently readout errors, whose mitigation with M3~\cite{nation2021scalable} would be excessively resource-intensive for large-scale problems. \rev{At the same time, noise is known to be highly detrimental for the performance of kernels, as it constitutes one of the sources of exponential concentration~\cite{thanasilp2024exponential}.}
\rev{A CVaR-based mitigation technique for fidelities, with provable bounds, has been proposed recently and applied to fidelity kernels on top of M3~\cite{barron_provable_2024}. It was shown to be effective in the error mitigation of fidelity quantum kernels (specifically, using ZZFeature map) up to 100 qubits and QAOA of up to 127 qubits, albeit with notable sampling overhead and processing time due to M3 error mitigation preprocessing.}

\rev{In this work,} we introduce a \rev{pragmatic} heuristic, called \textit{bit-flip tolerance} (BFT), that modifies the kernel computation from the standard expression in Eq.~\eqref{eq:fqk-align-classicaldata} to

\begin{equation}\label{eq:bitflip-tol-kernel}
    k^d_\lambda(x,x')=\sum_{\norm{j}_H \leq d}|\langle j|V_\lambda^{\dagger} D^{\dagger}(x) D(x') V_\lambda|0^n\rangle|^2
\end{equation}
where $\norm{j}_H$ is the Hamming weight of $j$ (namely, the number of ones in the bitstring representing $j$), and $d$ is the maximal tolerance. In another word, $d$ is the number of bits allowed to be flipped, compared to the ideal state $|0^n\rangle$ measured in Eq.~\eqref{eq:fqk-align-classicaldata}. The core idea is that minor deviations from the expected all-zero outcome, such as a few bits flipping to one, are likely attributable to noise rather than a true quantum effects. By reclassifying these near-zero strings as zeros, our approach aims to correct errors, thereby enhancing the robustness of the kernel computation.

However, this approach has its limitations. It inherently assumes that small deviations are always due to noise, which might not hold true in all situations. This assumption could lead to inaccuracies if the deviations contain meaningful information or if the errors are systematic or correlated rather than random. Additionally, a significant concern is that the resulting kernel function may not be positive semi-definite (PSD). To ensure PSD, we first have to check symmetry. In principle, it would be possible to define $k_S^d(x,x') := \frac{1}{2} k^d(x,x') + \frac{1}{2} k^d(x',x)$, where we dropped subscript $\lambda$ from Eq.~\eqref{eq:bitflip-tol-kernel}, but this approach nearly doubles the number of circuits that need to be evaluated.  Rather, we calculate the entries of the kernel matrix by $K^d_{ij} = K^d_{ji} = k^d(x_{\min \{i,j\}}, x_{\max \{i,j\}})$. Note that our definition depends on the sorting of data points, which is arbitrary. Now that symmetry is ensured, the kernel matrix may still not be positive semi-definite (PSD). To mitigate this issue, we project the non-PSD matrix onto the nearest PSD matrix. We also examine the relationship between the choice of $d$ and the distance of the matrix from the closest PSD, finding that larger $d$ values don't improve the positive semidefinite property of the matrix after a certain threshold.

It is important to note that this method is designed for fidelity quantum kernels and may not be generalizable. Additionally, the method tuning and the tests conducted in this manuscript are problem-specific, and require a systematic analysis of a wider spectrum of cases before asserting general applicability to quantum kernels. Nonetheless, to the best of our knowledge, this is the first proposed mitigation technique that specifically targets quantum kernels.

\subsection{Group structure of the embeddings in covariant kernels}\label{subsec:group}

In Subsec.~\ref{subsec:fidelity-qk} we recalled that the fiducial state $\ket\psi$ of a covariant kernel must be invariant under a subgroup of the unitary circuits. Hence, we find it convenient to formalize covariant quantum kernels as a tool to classify quantum unitaries, distancing ourselves from the common presentation of the topic.

Consider two classes of quantum operators $\mathcal{U}^+$ and $\mathcal{U}^-$, where $\mathcal{U}^\pm \subset U(2)^{\otimes n}$. Let $\ket\psi$ denote a fixed state-vector on $n$ qubits. For any $D, D' \in \mathcal{U} = \mathcal{U}^+ \cup \mathcal{U}^-$, the fidelity quantum kernel is defined as:
\begin{equation}\label{eq:fqk-operators}
k(D, D') = \abs{\braket{ \psi | D^\dag D' | \psi }}^2.
\end{equation}
Let us then define when a kernel is covariant with two classes of unitaries:
\begin{definition}\label{defn:group-structure} 
Let $\mathcal{U}^+$ and $\mathcal{U}^-$ be two classes in the space of quantum unitaries, namely two subsets of $U(2)^{\otimes n}$.
We say that the kernel in Eq.~\eqref{eq:fqk-operators} is covariant with the classes $\mathcal U^\pm$ of unitaries if 
there are a subgroup $\mathcal{S} < U(2)^{\otimes n}$ and two elements $C_\pm \in U(2)^{\otimes n}$ such that:
\begin{enumerate}
\item $\mathcal{U}^\pm$ are subsets of the cosets defined by $C_\pm$, namely $\mathcal{U}^\pm \subset C_\pm \mathcal{S}$,
\item $\ket\psi$ is $\mathcal{S}$-invariant modulo global phases, namely for all $S \in \mathcal{S}$ there exists one $\theta_S \in [0,2\pi)$ such that $S \ket\psi = e^{i \theta_S} \ket\psi$, and
\item $C_+ \ket\psi$ is orthogonal to $C_- \ket\psi$, namely $\braket{\psi | C_+^\dag C_- | \psi} = 0$.
\end{enumerate}
\end{definition}
The classification of the points $\{x_i\}$ can be translated into the classification of the embedding unitaries~$\{D(x_i)\}$, and the theory herein explained applies. In such case, the three properties in the definition are required in the space of the unitaries, and not necessarily in the original data. In other words, the needed structure is a characteristic of the embedding chosen for the data, which in turn must be suited to expose some structure of the classical data themselves. On the other side, following Ref.~\cite{glick_covariant_2024}, if the classical data belongs to a group $\mathcal G$, the first property of the previous definition may be inherited from a homomorphism from $\mathcal G$ to $U(2)^{\otimes n}$.
Still, the second and third properties must be verified in the space of quantum unitaries.

In contrast to the original work~\cite{glick_covariant_2024}, our Definition~\ref{defn:group-structure} relaxes the requirement of the $\mathcal{S}$-invariance of $\ket\psi$ by only necessitating it up to global phases.
Now, %
let
\begin{equation}\label{eq:deltaclass}
\delta^\mathrm{class}_{D, D'} :=
\begin{cases}
    1 &\text{if $D, D'$ in same class},\\
    0 &\text{otherwise.}
\end{cases}
\end{equation}
Under our Definition~\ref{defn:group-structure}, the following Proposition states that the condition $k = \delta^\mathrm{class}$ is actually \textit{equivalent to} the kernel being covariant with the unitary classes (and not only \textit{implied by} the kernel being covariant, as shown in Ref.~\cite{glick_covariant_2024}).

\begin{proposition}\label{prop:group}
Let $\mathcal{U}^+$ and $\mathcal{U}^-$ be two subsets of $U(2)^{\otimes n}$. A kernel $k$ is covariant with $\mathcal U^\pm$ if and only if $k(D, D') = \delta^\mathrm{class}_{D, D'}$ for all $D, D' \in \mathcal{U}^+ \cup \mathcal{U}^-$.
\end{proposition}

Definition~\ref{defn:group-structure} and Proposition~\ref{prop:group} straightforwardly extend to the multi-class case, see Appendix~\ref{appsub:covariant} where also the proof is given.

Proposition~\ref{prop:group} has significant implications: if the fidelity kernel matches the ideal kernel, it indicates the existence of the desired covariant structure in the unitaries, even if no structure is known in the original data. We introduce the term \emph{hidden} covariant structure in this case.
If no structure is known a priori, the fiducial state can be (approximately) identified~\cite{glick_covariant_2024} by resorting to the alignment of a parametric kernel as in Eq.~\eqref{eq:fqk-align-classicaldata}.

As introduced in Fig.~\ref{fig:featuremap}, we use Pauli rotations along one axis $\gamma$ ($\gamma \in \{X,Y,Z\}$) as the embedding layer. In this case, the structure of covariant kernels becomes very easy to analyze. Focusing for simplicity on a single qubit and mimicking the reasoning of Ref.~\cite{glick_covariant_2024}, one may take the subgroup $S := \{ R_\gamma(s \, \theta) \}_{s \in \mathbb Z}$, a subgroup of the Pauli rotations, where $\theta$ is a fixed angle (submultiple of $2 \pi$), and the coset generators $C_c := R_\gamma(\delta_c)$ for each class $c$. Additionally, $\delta_{c_1}-\delta_{c_2}$ should not be a multiple of $\theta$ for all classes $c_1 \neq c_2$, in order to ensure cosets are distinct. Therefore each data point $x$ would be of the form $x = s(x) \, \theta + \delta_{\mathrm{class}(x)}$, for some integer $s(x)$. 

\subsection{Synthetic data from the union of subspaces}\label{subsec:union-subspaces}

In our real-world dataset, increasing the number of features does not consistently enhance classification accuracy, particularly as we approach utility scale, where the number of features is 100+.  Therefore, to evaluate the robustness of the bit-flip tolerance technique at this scale, we turn to synthetic datasets. Since our primary focus is testing quantum kernels in real-world scenarios, we generate synthetic data based on the union of subspaces~\cite{vidal2011subspace, elhamifar2013sparse}. Indeed many practical datasets, including images, videos, and time series, exhibit properties consistent with the union of subspaces \cite{ elhamifar2013sparse, BahadoriKFL15, Khodadadzadeh7120510}. 

Data coming from a union of subspaces has an interesting inner product property, often exploited in tasks like clustering and classification \cite{Khodadadzadeh7120510, heckel2013subspace}. When data lies in independent subspaces %
of a given space $\mathbb R^n$, on average two points within the same subspace will have a higher inner product with each other compared to points from different subspaces. We apply the property to the case where the classes identify with the subspaces, and hence said property writes:%

\begin{equation}\label{eq:subspace-inner-classical}
\begin{aligned}
\mathbb{E}_{\mathrm{class}(x) \neq \mathrm{class}(x')}\left[ \left| \inner{x}{x'} \right|^2 \right] \\
< \mathbb{E}_{\mathrm{class}(x) = \mathrm{class}(x')}\left[ \left| \inner{x}{x'} \right|^2 \right].
\end{aligned}
\end{equation}
where $\mathbb{E}$ indicates the expectation value over all data points. 
The inequality holds, for instance, when $x$ and $x'$ are sampled uniformly from the unitary hypersphere and subspaces are independent, see Prop.~\ref{prop:subspace-classical} in Appendix~\ref{appsub:subspaces}.%

Moving to quantum kernels, %
it is required that the embedding preserves inequality~\eqref{eq:subspace-inner-classical}. More precisely, the property we are aiming for, is:

\begin{equation}\label{eq:subspace-inner}
\begin{aligned}
\mathbb{E}_{\mathrm{class}(x) \neq \mathrm{class}(x')}\left[ \left| \braket{\psi | D(x)^\dagger D(x') | \psi} \right|^2 \right] \\
< \mathbb{E}_{\mathrm{class}(x) = \mathrm{class}(x')}\left[ \left| \braket{\psi | D(x)^\dagger D(x') | \psi} \right|^2 \right],
\end{aligned}
\end{equation}
namely

\begin{equation*}
    \mathbb{E}_{\mathrm{class}(x) \neq \mathrm{class}(x')}\left[ k(x, x') \right] < \mathbb{E}_{\mathrm{class}(x) = \mathrm{class}(x')}\left[ k(x, x') \right].
\end{equation*}
Note that such condition constitutes a relaxed version of the covariant characterization $k(D(x), D(x')) = \delta^\mathrm{class}_{D(x), D(x')}$. The validity of Eq.~\eqref{eq:subspace-inner} depends on the embedding $D(\cdot)$ chosen for the data points. Indeed, take for instance $\ket\psi = \ket{0^n}$, and $D(x)$ to be a set of $Z$-rotations parametrized over the components of the data point $x$. In this case, it is easy to see that $k$ is identically~$1$ and the inequality does not hold. On the other hand, in Appendix~\ref{appsub:subspaces-quantum}, we show numerically that the inequality holds under very strong assumptions, namely when $\ket\psi = \ket{0^n}$, the encoding is made of $X$-rotations (or $Y$-rotations), the subspaces are independent, and the samples are uniformly drawn from the $2\pi$-hypersphere. This observation provides some confidence that fidelity kernels may have good performance on data generated from the union of subspaces. 

This relationship between quantum kernels and data coming from a union of subspaces is novel and relevant because it is known that clustering/multiclass classification of data which has this underlying structure is known to have instances which are computationally hard \cite{pesce2022subspace}.

\section{Application: Vehicle to Grid (V2G)}\label{sec:application}

A real-world application of multi-class QSVC is the Vehicle to Grid (V2G) problem as we discussed in Ref.~\cite{agliardi2024machine}. In short, the V2G problem consists of using electric vehicles as a Battery Energy Storage Solution (BESS) to buy and sell energy from the power grid. The optimization task is to decide whether an EV should charge or discharge at a given time, given factors such as the energy prices, the forecast energy demand, and the desired state of charge of electric vehicles (e.g., 60~kWh at 7~am). The problem is challenging due to its combinatorial nature and because charging and discharging decisions must be made online due to the time-varying system conditions. 

To make scheduling decisions fast, in our previous work~\cite{agliardi2024machine} we used a learning approach that combines offline simulation and online decision-making. In particular, we proposed to generate good candidate scheduling policies with approximate programming and then train a machine learning model to select a policy given the live information in the system. One way to make predictions is to use kernel methods. %

In more detail, in that approach, we designed a very fast approximate solver relying on a classifier (three class classification) that predicts the best action to be performed at a given time, among the following options: \texttt{C} (charge as many vehicles as feasible), \texttt{I} (keep all vehicles idle), \texttt{D} (discharge as many vehicles as feasible). Thus the problem is a three-class classification problem, where input features represent the state of the system of electric vehicles at time $k$: energy price, target volume to exchange with the grid, statistics of state of charge of vehicles (mean, variance, and higher moments, quantiles), state of constraints. %
Labels (\texttt{C}, \texttt{I}, or \texttt{D}) are assigned ahead of time by running an offline approximate solver, such as an approximate Dynamic Programming method or a time-bound execution of a commercial MILP solver.

\begin{figure}[t]
\centering
\includegraphics{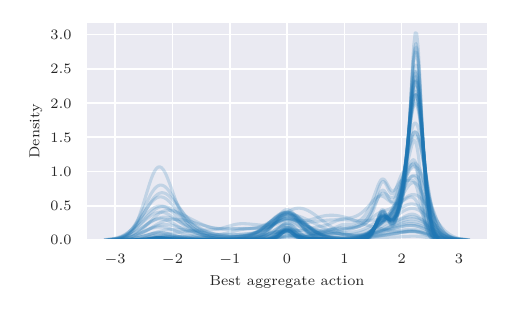}
\caption {
Density of the next best aggregate action grouped by time-to-go.
\label{fig:density}}
\end{figure}

Fig.~\ref{fig:density} shows that the best amount a fleet should be charged at a given time (the `next best aggregate action'), is a continuous quantity in nature, due to constraints. The figure indeed presents the density distribution of the next best aggregate action across different time horizons, normalized by the number of vehicles in each sample.
Despite the continuous domain, actions concentrate around three distinct peaks, corresponding to our three classes \texttt{C}, \texttt{I}, \texttt{D}~\cite{agliardi2024machine}. %

\section{Experiments}\label{sec:results}

In the earlier sections, we introduced the problem and provided an overview of the dataset, building on our previous work~\cite{agliardi2024machine}.

The dataset introduced in Sec.~\ref{sec:application} consists of $229{,}344$ samples and $178$ features. We run our experiments with a reduced set of features (of varying size), \rev{to adapt to the size of available quantum computers}. For feature selection, we sample $100$ random datasets and train random forest models, averaging the feature importance across runs. The top features align well with what a machine learning engineer might select. Although PCA could be an alternate way to reduce the dimensionality, we opt to preserve the original features for interpretability.

\begin{figure*}[t]
\centering
\includegraphics{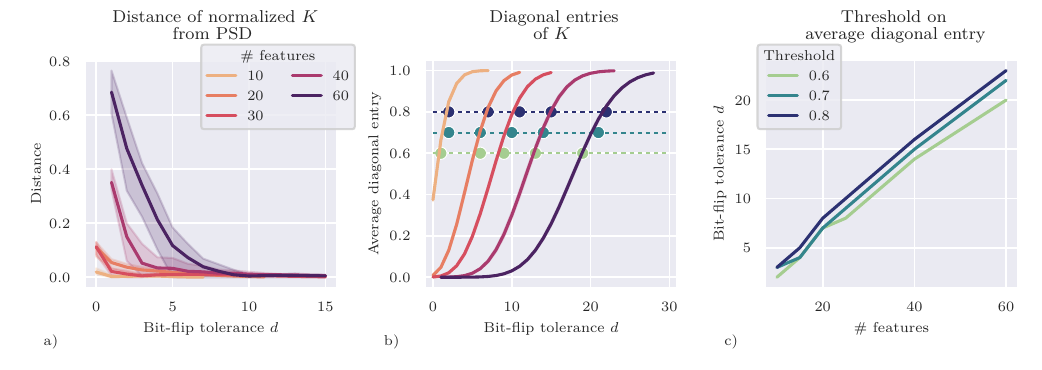}
\caption {
Hamming calibration on quantum hardware, with real data (average and $95\%$-confidence over 3 instances). Confidence area is imperceptible in b) and c). Color coding and legends apply across the various subplots. 
\textbf{a)} Distance of the normalized kernel matrix to the closer normalized PSD matrix, as a function of $d$.
\textbf{b)} Average diagonal entry in $K$. In absence of noise, this value should be $1$ independently of $d$. Dashed lines represent fixed thresholds for the average diagonal, and dots visualize the minimum number of bit-flips needed to satisfy such threshold.
\textbf{c)} Bit-flip tolerance $d$ required to satisfy a given threshold on the average diagonal entry of the kernel matrix.
\label{fig:ham-cal}}
\end{figure*}

\subsection{Calibration of bit-flip tolerance}\label{subsec:hamming-calibration}

As a first step towards hardware experiments, we need to estimate the value of $d$, which represents the maximal number of bits allowed to flip. Empirically, small values of~$d$ \rev{yield little deviation from standard fidelity kernels, and therefore lead to more accurate kernel evaluations, but are not sufficient to mitigate exponential concentration in presence of hardware noise. Conversely, larger values of~$d$ result in higher deviation and reduced kernel evaluation accuracy}. Finding an optimal value for~$d$ can be challenging and requires calibration experiments. 
For the calibration experiments we sample three random datasets, each with 15 samples and the number of features $[10, 20, 30, 40, 60]$. The selected dataset size for hardware runs is small, since the number of circuits needed for a single kernel evaluation scales quadratically with the number of samples. We evaluate the kernel matrices on the \texttt{ibm\_brussels} device, varying the bit-flip tolerance. In the variational form of Fig.~\ref{fig:featuremap}, we take $\alpha=X, \beta=Y$ and $\gamma=Z$. Experiments are conducted with $10{,}000$ shots. Dynamical decoupling~\cite{PhysRevApplied.20.064027} and gate twirling~\cite{PhysRevA.105.032620} were enabled.

As a first, useful observation for calibration, kernel matrices sampled from quantum computers may not always be PSD. Indeed, for $d=0$, it is known that in absence of noise and with infinite shots, the fidelity kernel is PSD, see Subsec.~\ref{subsec:fidelity-qk}. Despite this, it is easy to show empirically that in presence of noise and with finite measurements, the property does not always hold. On the other side, when $d$ equals the number of qubits, the kernel matrix is filled with ones, and it is again PSD. It is then important to assess the behavior for intermediate values of $d$ on noisy hardware. In Fig.~\ref{fig:ham-cal}~a), we show the distance of the normalized kernel matrix to the nearest normalized PSD matrix, obtained by clipping negative eigenvalues. The plot indicates that increasing the bit-flip tolerance from $0$ to a certain point leads to a significant reduction in this distance. Beyond this point, the distance starts flattening. This point can therefore be considered as a lower bound for selecting $d$.

The next, simple metric we examine is the average of the diagonal entries of the kernel matrices, which should ideally be equal to $1$. However, in the presence of noise and due to limitations of quantum kernels~\cite{thanasilp2024exponential}, diagonal entries deviate from $1$ in practice. Fig.~\ref{fig:ham-cal}~b) shows this metric, and it becomes clear that after a certain point, increasing the bit-flip tolerance does not significantly improve the diagonal values. This point can be considered an upper bound for choosing $d$. As witnessed by the imperceptibility of the confidence area in the plot, such point is independent of the problem instance. In fact, it is only a consequence of the hardware noise on fixed-structure circuits equivalent to the identity.

Again in the plot Fig.~\ref{fig:ham-cal}~b), we can fix a threshold (dashed line) and examine at which bit-flip tolerance the diagonal entries of the kernel matrices exceed this value, on average. The same data are more clearly visualized in the plot c) of same figure. Interestingly, we observe a linear dependency between the number of features and bit-flip tolerance for a fixed threshold. This discovery opens up a new direction for research, suggesting a link between number of feature and noise tolerance in quantum kernel evaluations.

\subsection{QKA on real data}\label{subsec:res-qka}
Now, let us evaluate how centralized quantum kernel alignment performs on the same real data. We start with noiseless simulations to then move to quantum hardware.

\begin{figure}[t]
\centering
\includegraphics{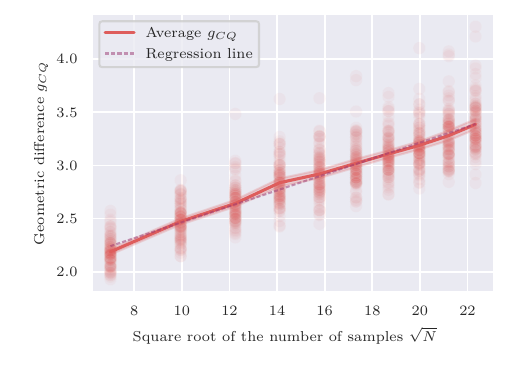}
\caption {
Geometric difference between the radial basis function (RBF) kernel and the aligned quantum kernel. Average and 95\% confidence interval over 100 instances for each $N$. Circles represent individual instances.
\label{fig:gcq}}
\end{figure}

As a preliminary test, we compute the geometric difference $g_{CQ}$~\cite{huang2021power} between radial basis function (RBF) kernels and the quantum kernels after alignment, see Fig.~\ref{fig:gcq}. Since RBF kernels are defined over a parameter $\gamma$, we choose the value of $\gamma$ minimizing $g_{CQ}$. The linear scaling of $g_{CQ}$ with the square root $\sqrt{N}$ of the number of train samples suggests potential quantum advantage~\cite{huang2021power}, motivating further experiments.

\begin{figure*}[t]
\centering
\includegraphics{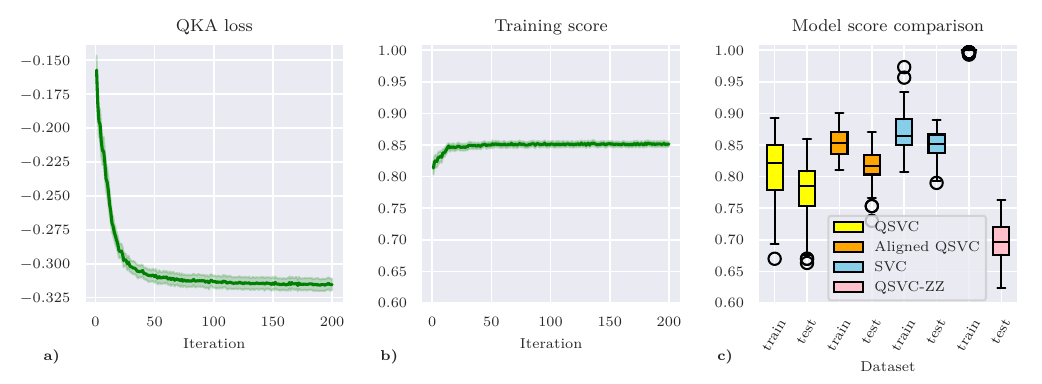}
\caption {
Simulation of QKA with real data, on 10 qubits and 600 samples. \textbf{a)}~Loss function during kernel alignment (average and $95\%$-confidence over 100 instances). \textbf{b)}~Train score during kernel alignment (average and $95\%$-confidence over 100 instances). The vertical axis is aligned with the next subplot for consistency. \textbf{c)}~Train and test scores of QKA with our feature map (marked as `QSVC' before alignment and as `Aligned QSVC' after alignment), compared to a classical SVC and to a QSVC with ZZ feature map (all box plots over 100 instances).
\label{fig:sim-result}}
\end{figure*}

\begin{figure*}[t]
\centering
\includegraphics{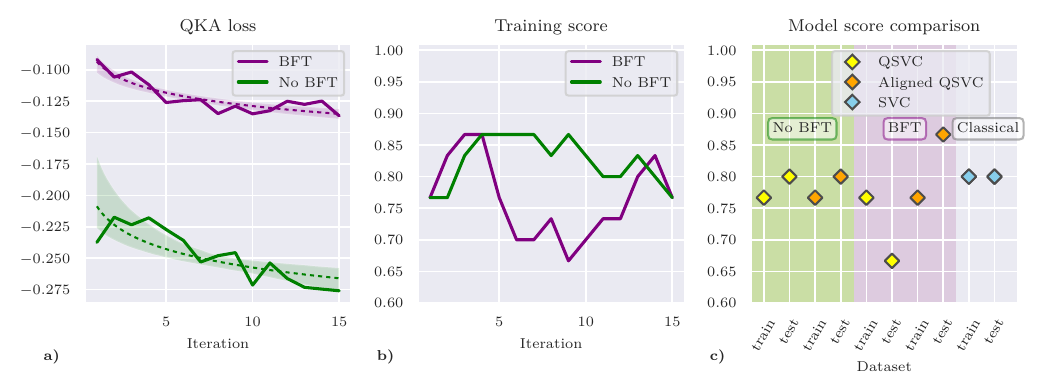}
\caption {
Hardware run of QKA with real data with and without bit-flip tolerance (BFT with $d=2$), on 10 qubits and 30 samples. Results with or without BFT are comparable at this size. \textbf{a)}~Loss function during kernel alignment (single instance), and logarithmic best-fit with $90\%$-confidence area. \textbf{b)}~Loss function during kernel alignment (single instance).  \textbf{c)}~Train and test scores of QKA compared to a classical SVC (single instance).
\label{fig:hardware-result}}
\end{figure*}

\begin{figure}[t]
\centering
\includegraphics{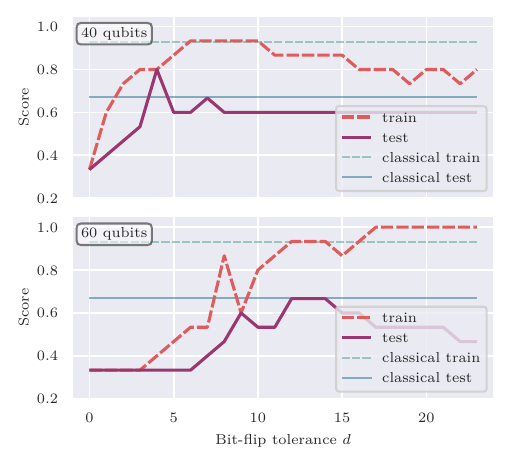}
\caption {
Train and test scores at different bit-flip tolerance levels, for hardware experiments with 40 and 60 qubits, compared to classical benchmarks. At this scale, the bit-flip tolerance method provides an improvement. The optimal tolerance~$d$ agrees with our calibration procedure on hardware.
\label{fig:hamming-scores}
}
\end{figure}

For the simulated results, we limit the setup to the $10$ most important features and, consequently, to circuits of $10$ qubits. We also set the bit-flip tolerance~$d$ to~$0$. This small configuration allows us to gather sufficient statistics while keeping the problem size reasonable from a machine learning perspective. We sub-sample $100$ random datasets, each containing $600$ samples, equally split into training and testing sets. Our QKA models are trained for $200$ iterations using SPSA as an optimizer. As a baseline, we use SVC from scikit-learn and apply a grid search of hyperparameters for a fair comparison. The results in Fig.~\ref{fig:sim-result} indicate that the quantum models perform nearly as well as classical SVC. To justify our choice of ansatz, we conducted several test. For example, let us remark that QSVC achieves far lower test scores with the commonly used ZZ feature map \cite{havlivcek2019supervised}, as shown in the same Fig.~\ref{fig:sim-result}~c).

As an additional benchmark, Projected Quantum Kernels~(PQK)~\cite{huang2021power} also provide results comparable with our aligned QSVC and with SVC. Specifically, the 1D Heisenberg feature map with random couplings~\cite{huang2021power} is applied, and projected back into the classical domain with $X$, $Y$, and $Z$ measurements. A classical RBF kernel is applied to such transformed features. We test $100$ datasets on simulator, each containing $600$ samples, equally split into training and testing sets. The train accuracy is $88\%$ and the test accuracy is $84\%$. 

Let us now move to the results on quantum devices. The datasets used in the simulation experiments are too large for current quantum hardware, due to the quadratic scaling mentioned earlier, as well as to the multiple iterations required by the kernel alignment optimization. Thus, for the hardware experiments, we sample a smaller dataset similar to the calibration experiments, limiting the number of features to~$10$ and reducing the dataset size to $30$ samples, equally split between training and testing data. The dataset is drawn from a pool where quantum models perform comparably well with the classical SVC, on a simulator. While this dataset size holds limited value from a machine learning perspective, the purpose here is to demonstrate the model trainability with bit-flip tolerance enabled. We acknowledge that the accuracy metrics may not generalize to larger models. Based on our calibration study, we set the bit-flip tolerance to $3$. The number of SPSA iterations is $15$. The remaining hardware setup mirrors the calibration experiments, including the presence of dynamical decoupling and measurement twirling. Qiskit reports the depth of the transpiled circuit as $23$. In Fig.~\ref{fig:hardware-result} we present two scenarios of quantum kernel alignment with BFT enabled (with tolerance $d=2$) and disabled. In both cases, the loss function follows the expected downward trend. The train score fluctuates more with BFT enabled but remains flat on average across iterations. Both models show performance comparable to the classical SVC proving viability of the BFT technique.

We next extend the setup to a larger number of qubits. Due to quantum hardware resource constraints, we further reduce the scope of the hardware experiments. Rather than aligning a kernel, we fit a single QSVC instance with arbitrary random parameters and we assess its performance across different bit-flip tolerance levels. For this experiment, we sample datasets with $40$ and $60$ features and corresponsing number of qubits, $30$ data points in each dataset, equally split into training and testing sets as usual. %
As shown in Fig.~\ref{fig:hamming-scores}, with no bit-flip tolerance, all train and test scores stand at $33\%$. Increasing the bit-flip tolerance improves these scores, as expected. The low test scores are attributed to the small dataset size as the models struggle to generalize on the limited number of samples. Indeed, they are in line with classical benchmarks. Nevertheless, the optimal bit-flip tolerance conforms with our previous calibration findings.

\begin{figure}[t]
\centering
\includegraphics{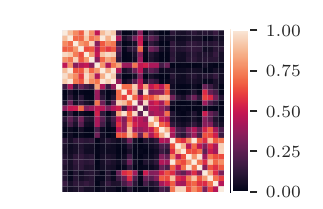}
\caption {
An example of kernel matrix for the training dataset, obtained after alignment on a simulator on 10 qubits and 30 samples.
\label{fig:block-structure}
}
\end{figure}

Connecting the empirical results above with the theory in Subsec.~\ref{subsec:group}, it is reasonable to expect that the good performance of the fidelity kernel is the result of the proper exploitation of some group structure in the data. An example of such a structure is shown in Fig.~\ref{fig:block-structure}. The raw data was further processed to filter only data points corresponding to the three peaks in Fig.~\ref{fig:density}, thus artificially reducing ambiguous and hard-to-classify data. Thence, the kernel matrix clearly shows a block diagonal structure representing the three distinct classes. Both quantum and classical models perform identically showing training and test accuracy of $100\%$ and $87\%$ respectively.

As a general consideration, it should be noted that with real data, one can only detect a covariant group structure up to some degrees of approximation, for multiple reasons:
\begin{itemize}
    \item training labels are noisy, what undermines the exact group structure at the origin,
    \item the data embedding chosen may not expose the data structure perfectly,
    \item the fiducial state may not be detected exactly, as an effect of exploring only a subspace of the Hilbert space by means of a variational form, or as an effect of a suboptimal parameter selection.
\end{itemize}
Consequently, $k(D_1, D_2)$ does not match $\delta^\mathrm{class}_{D_1, D_2}$ exactly, and the scores are less than~$1$. Despite the proof of Prop.~\ref{prop:group} provides a way to describe the group structure hidden in the data when the match is exact, no method to retrieve the underlying structure from the data is known for the non-ideal case. This observation offers a path for future developments.

\subsection{Experiments on synthetic datasets at utility scale}

\begin{figure*}[t]
\centering
\begin{minipage}{.49\textwidth}
    \includegraphics{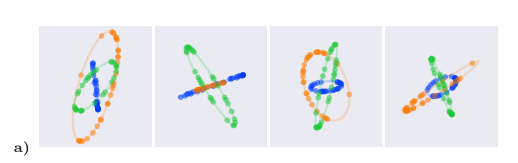}
\end{minipage}
\begin{minipage}{.49\textwidth}
    \includegraphics{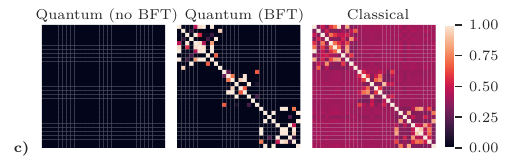}
\end{minipage}
\\[1mm]
\begin{minipage}{.49\textwidth}
    \includegraphics{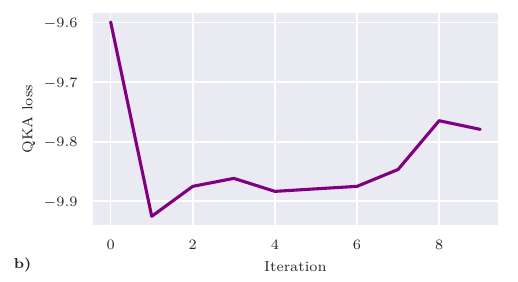}
\end{minipage}
\begin{minipage}{.49\textwidth}
    \includegraphics{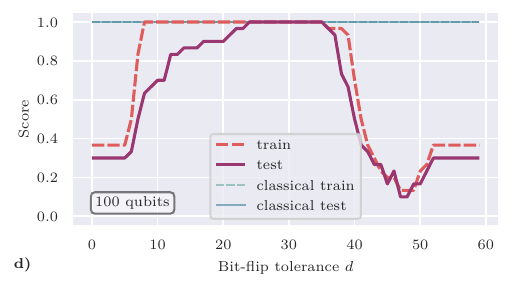}
\end{minipage}
\caption{
Hardware run of QKA with synthetic data, on 100 qubits.
\textbf{a)} Visualization of the dataset. Four projections of the same dataset on random pairs of orthogonal axes. Each color is a class.
\textbf{b)} Loss function during kernel alignment.
\textbf{c)} Comparison of kernel matrices for training datasets: quantum without BFT, quantum with BFT ($d=20$) after alignment, and classical RBF.
\textbf{d)} Train and test scores at different bit-flip tolerance levels, compared to classical benchmarks.
}
\label{fig:synthetic-datasets-hw100}%
\end{figure*}

\begin{figure}[t]%
\centering
\includegraphics{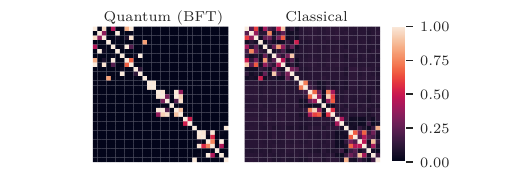}
\caption{
Hardware run of the fidelity kernel with synthetic data, on 156 qubits (no alignment). Comparison of kernel matrices for training datasets: quantum with BFT ($d=32$) and classical RBF (with alignment).
}
\label{fig:synthetic-datasets-hw156}%
\end{figure}

Let us now introduce the synthetic datasets generated from the union of subspaces, as described in Subsec.~\ref{subsec:union-subspaces}.
An initial classical testing shows improved performance with the following generalized RBF kernel: $k(x, x^\prime) = \gamma_1 \exp \left({- \frac{|x-x^\prime|^2}{2 \sigma_1^2}}\right) + \gamma_2 \exp\left({- \frac{|x-x^\prime|^2}{2 \sigma_2^2}}\right)$. Hence in this subsection we use it as the classical benchmark.
Preliminary results on simulator are contained in Appendix~\ref{appsub:sim-res}.
Moving to hardware experiments, we generate three 2d subspaces in a 100d space, as shown in Fig.~\ref{fig:synthetic-datasets-hw100} a). The dataset contains 60 samples, equally split between train and test. QKA is performed for 10 iterations with 500 shots, executed on the \texttt{ibm\_fez} device, with BFT at $d=20$. Referring to Fig.~\ref{fig:featuremap}, in the covariant feature map we take $\alpha=Z, \beta=Y$ in the fiducial state preparation layer, and $\gamma = X$ in the data embedding layer. All diagonal entries in the kernel are evaluated. The transpiled circuit depth for the kernel circuit is 58 \rev{and the two-qubit depth is 42}, with 900 RZ gates, 600 SX, and 198 CZ. The loss function during the alignment process is shown in Fig.~~\ref{fig:synthetic-datasets-hw100} b) where we observe that the loss function fluctuates as the fiducial state generator is learned during kernel alignment. Results are visualized in Fig.~\ref{fig:synthetic-datasets-hw100} c) and d). With thresholds $d$ of BFT between 24 and 35, the quantum test score of $100\%$ is retrieved, whereas without BFT, test score is $30\%$. The classical train and test scores are also $100\%$, respectively.

We also run a single experiment on 156 qubits on the newly released \texttt{ibm\_marrakesh}, without kernel alignment. In this case too, we consider three 2d subspaces. As before, shots are 500, and diagonal entries are evaluated. The transpiled circuit depth is 74 \rev{and the two-qubit depth is 58}, with 1404 RZ gates, 936 SX, and 310 CZ, see Fig.~\ref{fig:qc156} in the appendix for a visualization.  The bit flit tolerance (BFT) threshold is 32.
Fig.~\ref{fig:synthetic-datasets-hw156} shows that the quantum kernel matrix is similar to the classical kernel matrix.
Indeed, test scores are $80\%$ for the quantum model and $83\%$ for the classical one.

\section{Conclusions}\label{sec:conclusions}

In this work, we sought to comprehend the applicability of covariant quantum kernels, which are known to provide theoretical quantum advantage for specifically tailored datasets, to datasets of practical interest and at a utility scale. 
To this end, we applied fidelity quantum kernels to two multiclass classification problems: a real-world dataset related to optimizing the charge/discharge schedules of electric vehicles, as well as synthetic data generated from the union of subspaces. %
We established a theoretical connection between fidelity kernels and classical data belonging to the union of subspaces. We also demonstrated that fidelity quantum kernels can leverage an emergent property of such data. This paves the way for the application of fidelity kernels to the problem of clustering data with a subspace structure, a largely studied class of problems containing computationally hard instances~\cite{pesce2022subspace}.

To achieve our goal of conducting experiments at utility scale, we addressed \rev{one source of exponential concentration, namely hardware noise.} %
Exponential concentration is a common and very limiting problem, that can hinder generalization and degrade performance on unseen data.
We hence developed a readout mitigation strategy, termed bit-flip tolerance, specifically tailored to fidelity kernels. The technique is calibrated \emph{based on the ansatz and application}, in order to achieve reliable results on noisy hardware.
Remarkably, our ansatz was designed based on the quantum chip topology and hardware-native gates.
This approach enabled reliable experimentation with systems involving 100+ qubits.

Indeed, our experiments revealed that the fidelity quantum kernel achieves performance comparable to classical kernels on both real-world data, and on synthetic data at utility scale. Notably, both classical and quantum kernels demonstrated high accuracy, exceeding $80\%$, when applied to real-world data with $10$ features. Similarly, we achieved test scores of $80\%$ with $156$ features on synthetic data, matching the performance of classical kernels. Without BFT mitigation, however, the $156$-qubit experiments yielded an accuracy of only $37\%$. Combined with our other tests, these results strongly support the reliability of the proposed BFT strategy, even though the number of experiments was constrained by limited device availability. %
Notably, our $156$-qubit experiment represents the largest experimental demonstration of quantum machine learning on IBM hardware at the date of writing. We believe this work opens the door for future quantum machine learning experiments at scale, on real world data, which demonstrate some underlying structure that can be used for inductive bias model creation at regimes where classical computers cannot simulate.

\backmatter

\bmhead{Acknowledgements}
AD and GA are thankful to Nicola Mariella for fruitful discussions. CO thanks Jay Gambetta, Paul Nation, and Daniel Bultrini for useful discussions on covariant kernels and error mitigation. KY and GA thank Raja Hebbar and Heather Higgins for their leadership support. \rev{GA is also grateful to Jae-Eun Park, Brian Quanz, Chee-Kong Lee, Hwajung Kang, and Vaibhaw Kumar for valuable comments on the paper draft.}

\clearpage
\newgeometry{onecolumn}
\begin{appendices}

\renewcommand\thetheorem{\thesection\arabic{theorem}}

\section{Theoretical contributions}\label{app:proofs}
\setcounter{theorem}{0}

\subsection{Ideal kernels}\label{appsub:ideal}

\begin{proposition}[See Subsec.~\ref{subsec:multiclass-qka}]\label{prop:app:kernels}
Let $k_t$ and $k'_t$ be two (ideal) kernels defined by
$$k_t(D_1, D_2) :=
\begin{cases}
    1 &\text{if $D_1, D_2$ in same class},\\
    0 &\text{otherwise.}
\end{cases}$$
and
$$
    k'_t(D_1, D_2) =
\begin{cases}
    1 &\text{if $D_1, D_2$ in same class},\\
    -\frac{1}{C-1} &\text{otherwise,}
\end{cases}
$$
where $C$ is the number of classes. For a given dataset $\{x_i\}_{i=1}^m$, define the respective kernel matrices $K_t$ and $K'_t$, and their centered versions $K_t^c$ and $(K_t')^c$. The following holds: $\frac{K_t^c}{\norm{K_t^c}} = \frac{(K_t')^c}{\norm{(K_t')^c}}$, and hence $\mathcal{A}(K_t, \cdot) = \mathcal{A}(K_t', \cdot)$.
\end{proposition}
\begin{proof}
Let us drop the subscript $t$ for convenience. Observe that $K = \frac{K' - a \mathbf{1}}{1-a}$ where $\mathbf{1}$ is a matrix of ones, and $a=-\frac{1}{C-1}$.
It is obvious that $\frac{K^c}{\norm{K^c}}$ in invariant under multiplicative constants, so we only need to show that it is invariant under translations. Thence, let $K'' = K + \alpha \mathbf1$, for some $\alpha$. Recalling the definition of centered kernel,
\begin{equation*}
\begin{split}
(K'')^c & = K'' - \mathbf1 K''/m - K'' \mathbf1/m + \mathbf1 K'' \mathbf1/m^2\\
& = (K + \alpha \mathbf1) - \mathbf1 (K + \alpha \mathbf1)/m - (K + \alpha \mathbf1) \mathbf1/m + \mathbf1 (K + \alpha \mathbf1) \mathbf1/m^2\\
& = K^c + \alpha \mathbf1 - \alpha \mathbf1 \mathbf1 /m - \alpha \mathbf1 \mathbf1/m + \alpha \mathbf1 \mathbf1 \mathbf1/m^2\\
& = K^c + \alpha \mathbf1 - \alpha \mathbf1 - \alpha \mathbf1 + \alpha \mathbf1\\
& = K^c.
\end{split}
\end{equation*}
\end{proof}

\subsection{Covariant kernels}\label{appsub:covariant}

\begin{proposition}[Restatement of Prop.~\ref{prop:group}]
Let $\ket{\psi}$ be a given statevector, and let $\{ \mathcal{U}_j \}_{j=1}^m$ be subsets of $U(2)^{\otimes n}$. Then, the fidelity quantum kernel $k(D, D')$ defined as
$$k(D, D') = \abs{\braket{ \psi | D^\dag D' | \psi }}^2$$
in Eq.~\eqref{eq:fqk-operators} equals $\delta^\mathrm{class}_{D, D'}$ defined in Eq.~\eqref{eq:deltaclass} for all $D, D' \in \bigcup_{j=1}^m \mathcal{U}_j$ if and only if $k$ is covariant with $\{ \mathcal{U}_j \}_{j=1}^m$, namely if and only if there are a subgroup $\mathcal{S} < U(2)^{\otimes n}$ and elements $C_j \in U(2)^{\otimes n}$ for $j=1,...,m$ such that:
\begin{enumerate}
\item $\mathcal{U}_j$ is a subsets of the coset defined by $C_j$, namely $\mathcal{U}_j \subset C_j \mathcal{S}$,
\item $\ket\psi$ is $\mathcal{S}$-invariant modulo global phases, namely for all $S \in \mathcal{S}$ there exists one $\theta_S \in [0,2\pi)$ such that $S \ket\psi = e^{i \theta_S} \ket\psi$, and
\item for all $\ell \neq j$, $C_j \ket\psi$ is orthogonal to $C_\ell \ket\psi$, namely $\braket{\psi | C_j^\dag C_\ell | \psi} = 0$.
\end{enumerate}
\end{proposition}

\begin{proof}
Assume covariance. Take $D \in \mathcal U_j$ and $D' \in \mathcal U_\ell$. Write
\begin{eqnarray*}
    \braket{\psi | D^\dag D' | \psi} &=& \braket{\psi | S^\dag C_j^\dag C_\ell S' | \psi}, \quad \text{some $S,S' \in \mathcal{S}$, by Property 2}\\
    &=& e^{- i \theta + i \theta'} \braket{\psi | C_j^\dag C_\ell | \psi}, \quad \text{some $\theta, \theta'$, by Property 2}\\
    &=& e^{- i \theta + i \theta'} \delta_{j = \ell}, \quad \text{by Property 3.}
\end{eqnarray*}
Therefore $k(D, D') = \abs{ e^{- i \theta + i \theta'} \delta_{j = \ell} }^2 = \delta^\mathrm{class}_{D, D'}$.

Vice versa, assume $k(D, D') = \delta^\mathrm{class}_{D, D'}$. For all $j=1,...,m$, fix one element $C_j$ in $\mathcal{U}_j$. Then, let $\mathcal{S}$ be the  group generated by $C_j^\dag \mathcal{U}^j$ for all $j$, namely
$$\mathcal{S}:= \subgrgen { \bigcup_{j=1}^m \left\{ C_j^\dag D_j \right\}_{D_j \in \mathcal{U}_1} }.$$
Property 1 of the covariant structure holds by definition of $\mathcal{S}$, since any $D \in \mathcal{U}_j$ can be written as $D=C_j C_j^\dag D$ with $C_j^\dag D \in \mathcal{S}$.

Additionally, we know that $1 = k(D, D') = \abs{\braket{ \psi | D^\dag D' | \psi}}^2 = 1$ for all $D, D' \in \mathcal{U}_j$. Specifically, $\abs{\braket{ \psi | C_j^\dag D | \psi}}^2 = 1$ for all $D \in \mathcal{U}_j$. Therefore $C_j^\dag D \ket\psi = e^{i \theta} \ket\psi$, for all $D \in \mathcal{U}^+$ and some $\theta = \theta(D)$. Consequently $\ket\psi$ is invariant under all generators of $\mathcal{S}$ modulo global phases, and hence under $\mathcal{S}$, thus proving Property~2.

Lastly, Property~3 directly follows from the hypothesis that $k(D, D') = 0$ for all $D \in \mathcal{U}_j$ and $D' \in \mathcal{U}_\ell$, by taking $D = C_j$ and $D' = C_\ell$.
\end{proof}

\subsection{Union of subspaces: classical inequalities}\label{appsub:subspaces}
This subsection culminates in Prop.~\ref{prop:subspace-classical}, stating that the average inner product of two elements from the same subspace is higher than the average inner product between vectors in distinct subspaces, under appropriate conditions. In order to prove it, some preliminaries are introduced below.

\begin{lemma}\label{lem:sqrt-inn-prod}
The expected squared inner product between a uniformly sampled point on the unitary sphere $S^{d-1}$ and a fixed point on the same sphere is $1/d$. In formulas:
$$\frac{1}{\abs{S^{d-1}}} \int_{S^{d-1}} \abs{\inner{x}{\bar x}}^2 \, \dd S^{d-1}(x) = \frac{1}{d}$$
\end{lemma}
\begin{proof}
By symmetry, we can assume $\bar x = (1,0, ..., 0)$.
In hyperspherical coordinates, and writing $\abs{S^{d-1}}$ explicitely, the left hand side $J$ of the thesis rewrites:
$$J = \frac
{\int_0^\pi \dd \theta_1 \int_0^\pi \dd \theta_2 \cdots \int_0^{2\pi} \dd \theta_{d-1} \cos^2 \theta_1 \sin^{d-2} \theta_1 \sin^{n-3} \theta_2 \cdots \sin \theta_{n-2}}
{\int_0^\pi \dd \theta_1 \int_0^\pi \dd \theta_2 \cdots \int_0^{2\pi} \dd \theta_{d-1} \sin^{d-2} \theta_1 \sin^{n-3} \theta_2 \cdots \sin \theta_{n-2}}.
$$
By separating the variables:
$$J = \frac
{\int_0^\pi \cos^2 \theta_1 \sin^{d-2} \theta_1 \, \dd \theta_1}
{\int_0^\pi \sin^{d-2} \theta_1 \, \dd \theta_1}
.$$
Now, set $I_d := \int_0^\pi \sin^{d} \theta_1 \, \dd \theta_1$. Writing $\cos^2 \theta_1 = 1- \sin^2 \theta_1$, the last formula becomes
$$J = \frac{ I_{d-2} - I_d }{I_{d-2}} = 1-\frac{I_d}{I_{d-2}} .$$
Now, $I_d$ can be integrated by parts, by differentiating the factor $\sin \theta$, to show that $I_d = (d-1) I_{d-2} - (d-1) I_d$, namely $I_d = \frac{d-1}{d} I_{d-2}$. Thence,
$$J = 1- \frac{I_d}{I_{d-2}} = 1- \frac{d-1}{d} = \frac{1}{d}.$$
This completes the proof.
\end{proof}

\begin{definition}[Principal angles~\cite{heckel2015robust}] Let $X$ and $Y$ be two $d$-dimensional subspaces of $\mathbb R^D$. Define recursively $\{ u_1, ..., u_d \}$ and $\{ v_1, ..., v_d \}$ to be orthonormal bases of $X$ and $Y$ resp., by
$$u_j, v_j := {\arg\max} \inner{x}{y},$$
where the $\arg\max$ is taken on $x \in X$, $y \in Y$, $\norm{x}_2 = \norm{y}_2 = 1$, $x \bot u_i$ and $y \bot v_i$ for all $i=1,...,j$. Then, the principal angles are defined by
$$\cos \theta_j := \inner{u_j}{v_j}.$$
\end{definition}

\begin{lemma}\label{lem:orth}
Let $\{ u_1, ..., u_d \}$ and $\{ v_1, ..., v_d \}$ be one pair of orthonormal bases defining the principal angles. Then, $u_i \bot v_j$ for all $i \neq j$.
\end{lemma}
\begin{proof}
Assume $u_i$ not orthogonal to $v_j$, for some $i<j$. Let $y := \frac{\inner{u_i}{v_i} v_i + \inner{u_i}{v_j} v_j}{\sqrt{\inner{u_i}{v_i}^2 + \inner{u_i}{v_j}^2}}$. One can easily show that $y$ is normal and that $\inner{u_i}{y}$ is greater than $\inner{u_i}{v_i}$, contradicting the definition of $v_i$.
\end{proof}

\begin{proposition}\label{prop:subspace-classical}
Let $X, Y$ be two linearly independent $d$-dimensional subspaces of~$\mathbb R^n$. Let $x, x'$ be uniformly and independently sampled from the unit sphere in $X$, and $y$ be uniformly sampled from the unit sphere in $Y$ independently of $x, x'$. Then,
$$\mathbb E [\abs{\inner{x}{y}}^2] < \mathbb E [\abs{\inner{x}{x'}}^2].$$
\end{proposition}
\begin{proof}
For the left hand side, write
\begin{eqnarray*}
\mathbb E [\abs{\inner{x}{y}}^2]
&=& \frac{1}{\abs{S^{d-1}}^2} \int_{S^{d-1}} \dd S^{d-1}(x) \int_{S^{d-1}} \dd S^{d-1}(y) \, \abs{\inner{x}{y}}^2 \\
&=& \frac{1}{\abs{S^{d-1}}^2} \int_{S^{d-1}} \dd S^{d-1}(x) \int_{S^{d-1}} \dd S^{d-1}(y) \, \abs*{\inner{ \sum_{i=1}^d \inner{x}{u_i} u_i }{\sum_{j=1}^d \inner{y}{v_j} v_j}}^2 \\
&=& \frac{1}{\abs{S^{d-1}}^2} \int_{S^{d-1}} \dd S^{d-1}(x) \int_{S^{d-1}} \dd S^{d-1}(y) \, \abs*{\sum_{i,j} \inner{x}{u_i} \inner{y}{v_j} \inner{u_i}{v_j}}^2 \\
&=& \frac{1}{\abs{S^{d-1}}^2} \int_{S^{d-1}} \dd S^{d-1}(x) \int_{S^{d-1}} \dd S^{d-1}(y) \, \abs*{\sum_{j=1}^d \inner{x}{u_j} \inner{y}{v_j} \inner{u_j}{v_j}}^2,
\end{eqnarray*}
where we used Lemma~\ref{lem:orth} in the last step. Therefore,
\begin{eqnarray*}
\mathbb E [\abs{\inner{x}{y}}^2]
&\leq& \frac{1}{\abs{S^{d-1}}^2} \int_{S^{d-1}} \dd S^{d-1}(x) \int_{S^{d-1}} \dd S^{d-1}(y) \, \sum_{j=1}^d \abs{\inner{x}{u_j} \inner{y}{v_j} \inner{u_j}{v_j}}^2 \\
&=& \frac{1}{\abs{S^{d-1}}^2} \sum_{j=1}^d \abs{\inner{u_j}{v_j}}^2 \int_{S^{d-1}} \abs{\inner{x}{u_j}}^2 \, \dd S^{d-1}(x) \int_{S^{d-1}} \abs{\inner{y}{v_j} }^2 \, \dd S^{d-1}(y) \\
&=& \frac{1}{d^2} \sum_{j=1}^d \cos^2 \theta_j,
\end{eqnarray*}
by Lemma~\ref{lem:sqrt-inn-prod} and by the definition of principal angles $\theta_j$. Clearly, the last expression is smaller than $1/d$, that is the right hand side in the thesis, again by Lemma~\ref{lem:sqrt-inn-prod}.
\end{proof}

\subsection{Union of subspaces: quantum inequalities}\label{appsub:subspaces-quantum}

\begin{figure}[t]%
\centering
\includegraphics{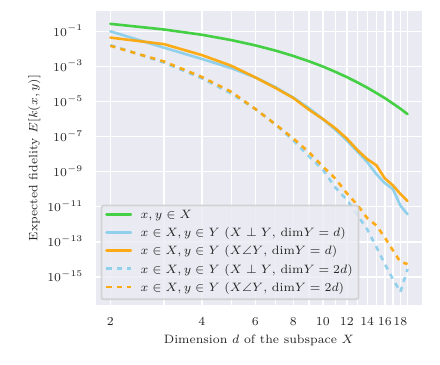}
\caption{
The expectation of $k(x,y)$, as a function of the subspace size, in five different cases: (i) $x, y$ belong to the same subspace $X$, (ii) $x \in X, y \in Y$, with $X$ and $Y$ orthogonal an equally sized, (iii) $x \in X, y \in Y$, with $X$ and $Y$ independent non-orthogonal, and $\dim Y = \dim X$, (iv) $x \in X, y \in Y$, with $X$ and $Y$ orthogonal, and $\dim Y = 2 \dim X$, (v) $x \in X, y \in Y$, with $X$ and $Y$ independent non-orthogonal, and $\dim Y = 2 \dim X$.
}
\label{fig:subspace-quantum}%
\end{figure}

The objective of this subsection is to show numerically that an analog of Prop.~\ref{prop:subspace-classical} holds for the quantum case, i.e. once we replace the classical kernel entries $k(x,y) = \abs{\inner{x}{y}}^2$, with the quantum kernel entries $k(x,y) = \abs{\braket{\psi | D(x)^\dagger D(y) | \psi} }^2$, under some additional hypotheses.

Like above, let $X,Y$ be two linearly independent subspaces of $\mathbb R^n$, and assume $x$ to lie on the unitary sphere. We put ourselves in the simplified setting of $\ket\psi = \ket 0$, and $D(x)$ to be the $R_X$ angle encoding, namely $D(x) = \bigotimes_{i=1}^n R_X(2 \pi x_i)$. Then, $k(x,y) = \abs{\braket{\psi | D(x)^\dagger D(y) | \psi} }^2 = \prod_{i=1}^n \cos^2(x_i-y_i)$. The same would hold for the $R_Y$ angle encoding. Now, Fig.~\ref{fig:subspace-quantum} shows the expectation of $k(x,y)$, estimated via $10{,}000{,}000$ random shots, as a function of the subspace size, when the two points belong to a same subspace, or when they belong to $X$ and $Y$, varying the mutual position of $X$ and $Y$ (orthogonal or independent non-orthogonal), as well as the dimension of $Y$ compared to that of $X$. Data for the experiments are generated exploiting the fact that, taken normally distributed vector components in $\mathbb R^d$, the normalized vector is uniformly distributed on the sphere $S^{d-1}$. Starting from independent orthogonal subspaces, the data for independent non-orthogonal subspaces are obtained by applying a random rotation to $Y$ in the space $\mathbb R^{\dim X + \dim Y}$, see \texttt{special\_ortho\_group.rvs()} in Scipy~\cite{2020SciPy-NMeth}.

From the figure, it emerges that the desired inequality~\eqref{eq:subspace-inner} holds in the tested settings, as the green line is always above the others.

\section{Additional experimental information}

\subsection{Simulator results on synthetic data}\label{appsub:sim-res}

\begin{figure}[t]%
\centering
\includegraphics{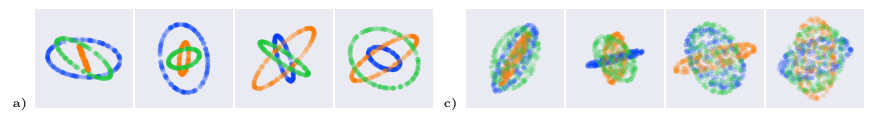}
\\[1mm]
\includegraphics{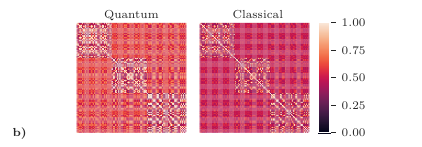}
\includegraphics{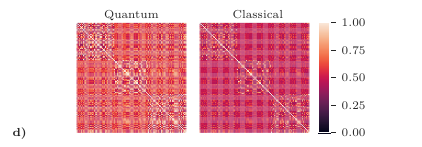}
\caption{
Simulation of QKA with synthetic data, on 10 qubits.
\textbf{a)} Visualization of the data. Here, classes are 2d subspaces of the 10d space. Four projections of the same dataset on random pairs of orthogonal axes. Each color is a class.
\textbf{b)} Kernel matrix for the training dataset after alignment, for the 2d case.
\textbf{c)} Visualization of the data. Here, classes are 3d subspaces of the 10d space.
\textbf{d)} Kernel matrix for the training dataset after alignment, for the 3d case.
}
\label{fig:synthetic-data-kernels-sim}%
\end{figure}

The dataset has a total vector space dimension of $10$, and $3$ random linearly independent $2$d or $3$d vector subspaces, see Fig.~\ref{fig:synthetic-data-kernels-sim} a) and c) respectively. For each class, $200$ data points are uniformly sampled from the unit sphere. The dataset is evenly split into train and test sets. We run 5 SPSA steps for $2$d and 55 steps for $3$d and utilize qiskit statevector simulator. The clear diagonal block structure in the kernel matrix (Fig.~\ref{fig:synthetic-data-kernels-sim}~b)~and~d)) motivates the interest in data generated from the union of subspaces. The $2$d vector subspace quantum and classical test set accuracies are $100\%$, whereas the $3$d subspace quantum test accuracy is $88\%$ and classical test accuracy is $100\%$.

\subsection{Hardware results on synthetic data}\label{appsub:hw-res}

Fig.~\ref{fig:qc156} visualizes the 156-qubit kernel circuit.

\begin{figure}%
\centering
\includegraphics[width=\textwidth]{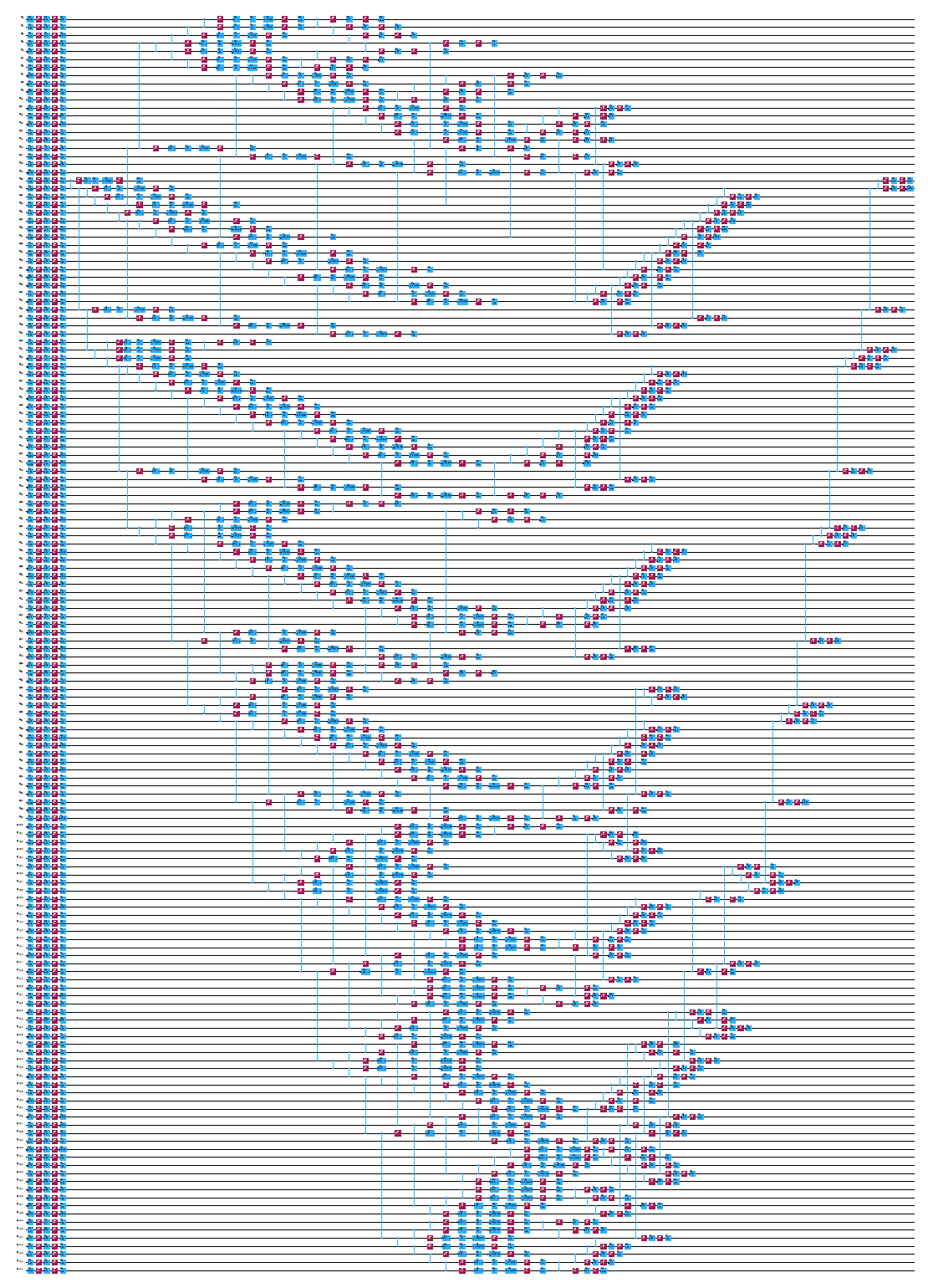}
\caption{
The quantum circuit that evaluates the 156-qubit kernel. }
\label{fig:qc156}%
\end{figure}

\subsection{Hardware specification of the IBM device}\label{appsub:device specifications}

A topology graph of IBM Marrakesh and some of its hardware specifications are provided in Fig.~\ref{fig:marrakesh}. %
\begin{figure}
	\centering
	\includegraphics[width=.7\linewidth]{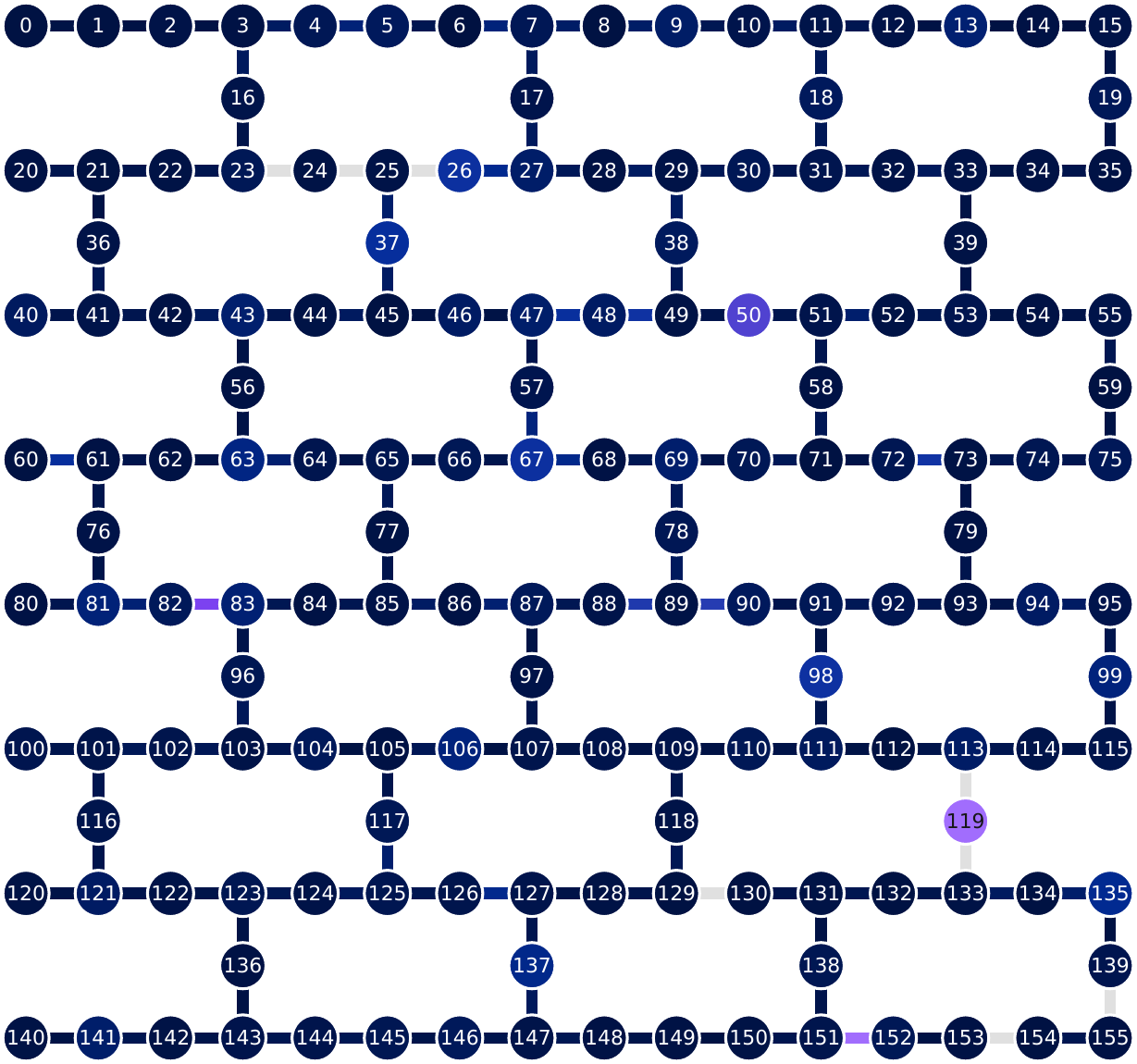}
    \\[2mm]
	\caption{The qubit topology graph of IBM Marrakesh, which is a 156-qubit quantum device. Lighter colors on qubits represent higher readout assignment errors, while on edges represent higher $CZ$ errors. This device contains a Heron r2 processor with basis gates $CZ$, $ID$, $RX$, $RZ$, $RZZ$, $SX$, and $X$ respectively. This quantum chip has median $CZ$ error of $2.281 \cdot 10^{-3}$, median $SX$ error $2.664 \cdot 10^{-4}$, median readout error $1.440 \cdot 10^{-2}$, median T1 is $175.11$ $\mu$s, and median T2 is $108.26$ $\mu$s respectively.}
	\label{fig:marrakesh}
\end{figure}

\end{appendices}

\clearpage
\bibliographystyle{unsrt}
\bibliography{references}%

\begin{thebibliography}{10}

\bibitem{boser1992training}
Bernhard~E Boser, Isabelle~M Guyon, and Vladimir~N Vapnik.
\newblock A training algorithm for optimal margin classifiers.
\newblock In {\em Proceedings of the fifth annual workshop on Computational learning theory}, pages 144--152, 1992.

\bibitem{scholkopf2018learning}
Bernhard Scholkopf and Alexander~J Smola.
\newblock {\em Learning with kernels: support vector machines, regularization, optimization, and beyond}.
\newblock MIT press, 2018.

\bibitem{meijering1999image}
Erik~HW Meijering, Karel~J Zuiderveld, and Max~A Viergever.
\newblock Image reconstruction by convolution with symmetrical piecewise nth-order polynomial kernels.
\newblock {\em IEEE transactions on image processing}, 8(2):192--201, 1999.

\bibitem{PhysRevLett.113.130503}
Patrick Rebentrost, Masoud Mohseni, and Seth Lloyd.
\newblock Quantum support vector machine for big data classification.
\newblock {\em Phys. Rev. Lett.}, 113:130503, Sep 2014.

\bibitem{huang2021power}
Hsin-Yuan Huang, Michael Broughton, Masoud Mohseni, Ryan Babbush, Sergio Boixo, Hartmut Neven, and Jarrod~R McClean.
\newblock Power of data in quantum machine learning.
\newblock {\em Nature communications}, 12(1):2631, 2021.

\bibitem{havlivcek2019supervised}
Vojt{\v{e}}ch Havl{\'\i}{\v{c}}ek, Antonio~D C{\'o}rcoles, Kristan Temme, Aram~W Harrow, Abhinav Kandala, Jerry~M Chow, and Jay~M Gambetta.
\newblock Supervised learning with quantum-enhanced feature spaces.
\newblock {\em Nature}, 567(7747):209--212, 2019.

\bibitem{liu2021rigorous}
Yunchao Liu, Srinivasan Arunachalam, and Kristan Temme.
\newblock A rigorous and robust quantum speed-up in supervised machine learning.
\newblock {\em Nature Physics}, 17(9):1013--1017, 2021.

\bibitem{slattery_numerical_2023}
Lucas Slattery, Ruslan Shaydulin, Shouvanik Chakrabarti, Marco Pistoia, Sami Khairy, and Stefan~M. Wild.
\newblock Numerical evidence against advantage with quantum fidelity kernels on classical data.
\newblock 107(6):062417.

\bibitem{krunic_quantum_2022}
Zoran Krunic, Frederik Flother, George Seegan, Nate Earnest-Noble, and Shehab Omar.
\newblock Quantum kernels for real-world predictions based on electronic health records.
\newblock 3:1--11.

\bibitem{glick_covariant_2024}
Jennifer~R. Glick, Tanvi~P. Gujarati, Antonio~D. Córcoles, Youngseok Kim, Abhinav Kandala, Jay~M. Gambetta, and Kristan Temme.
\newblock Covariant quantum kernels for data with group structure.
\newblock {\em Nature Physics}, 20(3):479--483, 2024.

\bibitem{thanasilp2024exponential}
Supanut Thanasilp, Samson Wang, M~Cerezo, and Zo{\"e} Holmes.
\newblock Exponential concentration in quantum kernel methods.
\newblock {\em Nature communications}, 15(1):5200, 2024.

\bibitem{agliardi2024machine}
Gabriele Agliardi, Giorgio Cortiana, Anton Dekusar, Kumar Ghosh, Naeimeh Mohseni, Corey O'Meara, V{\'\i}ctor Valls, Kavitha Yogaraj, and Sergiy Zhuk.
\newblock A machine learning approach to boost the vehicle-2-grid scheduling.
\newblock {\em arXiv preprint arXiv:2407.20802}, 2024.

\bibitem{elhamifar2013sparse}
Ehsan Elhamifar and Ren{\'e} Vidal.
\newblock Sparse subspace clustering: Algorithm, theory, and applications.
\newblock {\em IEEE transactions on pattern analysis and machine intelligence}, 35(11):2765--2781, 2013.

\bibitem{BahadoriKFL15}
Mohammad~Taha Bahadori, David~C. Kale, Yingying Fan, and Yan Liu.
\newblock Functional subspace clustering with application to time series.
\newblock In Francis~R. Bach and David~M. Blei, editors, {\em Proceedings of the 32nd International Conference on Machine Learning, {ICML} 2015, Lille, France, 6-11 July 2015}, volume~37 of {\em {JMLR} Workshop and Conference Proceedings}, pages 228--237. JMLR.org, 2015.

\bibitem{Khodadadzadeh7120510}
Mahdi Khodadadzadeh, Jun Li, Antonio Plaza, and José~M. Bioucas-Dias.
\newblock Hyperspectral image classification based on union of subspaces.
\newblock In {\em 2015 Joint Urban Remote Sensing Event (JURSE)}, pages 1--4, 2015.

\bibitem{vidal2011subspace}
Ren{\'e} Vidal.
\newblock Subspace clustering.
\newblock {\em IEEE Signal Processing Magazine}, 28(2):52--68, 2011.

\bibitem{mohri_foundations_2018}
Mehryar Mohri, Afshin Rostamizadeh, and Ameet Talwalkar.
\newblock {\em Foundations of machine learning}.
\newblock Adaptive computation and machine learning. The {MIT} Press, second edition edition, 2018.

\bibitem{shawe-taylor_kernel_2004}
John Shawe-Taylor and Nello Cristianini.
\newblock {\em Kernel methods for pattern analysis}.
\newblock Cambridge University Press, 2004.
\newblock {OCLC}: 144618454.

\bibitem{cortes2012algorithms}
Corinna Cortes, Mehryar Mohri, and Afshin Rostamizadeh.
\newblock Algorithms for learning kernels based on centered alignment.
\newblock {\em The Journal of Machine Learning Research}, 13:795--828, 2012.

\bibitem{cristianini2001kernel}
Nello Cristianini, John Shawe-Taylor, Andre Elisseeff, and Jaz Kandola.
\newblock On kernel-target alignment.
\newblock {\em Advances in neural information processing systems}, 14, 2001.

\bibitem{Camargo2009}
Jorge~E. Camargo and Fabio~A. Gonz{\'a}lez.
\newblock A multi-class kernel alignment method for image collection summarization.
\newblock In Eduardo Bayro-Corrochano and Jan-Olof Eklundh, editors, {\em Progress in Pattern Recognition, Image Analysis, Computer Vision, and Applications}, pages 545--552, Berlin, Heidelberg, 2009. Springer Berlin Heidelberg.

\bibitem{PhysRevApplied.20.064027}
Nic Ezzell, Bibek Pokharel, Lina Tewala, Gregory Quiroz, and Daniel~A. Lidar.
\newblock Dynamical decoupling for superconducting qubits: A performance survey.
\newblock {\em Phys. Rev. Appl.}, 20:064027, Dec 2023.

\bibitem{qiskit2024}
Ali Javadi-Abhari, Matthew Treinish, Kevin Krsulich, Christopher~J. Wood, Jake Lishman, Julien Gacon, Simon Martiel, Paul~D. Nation, Lev~S. Bishop, Andrew~W. Cross, Blake~R. Johnson, and Jay~M. Gambetta.
\newblock Quantum computing with {Q}iskit, 2024.

\bibitem{PhysRevA.105.032620}
Ewout van~den Berg, Zlatko~K. Minev, and Kristan Temme.
\newblock Model-free readout-error mitigation for quantum expectation values.
\newblock {\em Phys. Rev. A}, 105:032620, Mar 2022.

\bibitem{nation2021scalable}
Paul~D Nation, Hwajung Kang, Neereja Sundaresan, and Jay~M Gambetta.
\newblock Scalable mitigation of measurement errors on quantum computers.
\newblock {\em PRX Quantum}, 2(4):040326, 2021.

\bibitem{barron_provable_2024}
Samantha~V. Barron, Daniel~J. Egger, Elijah Pelofske, Andreas Bärtschi, Stephan Eidenbenz, Matthis Lehmkuehler, and Stefan Woerner.
\newblock Provable bounds for noise-free expectation values computed from noisy samples.
\newblock 4(11):865--875.

\bibitem{heckel2013subspace}
Reinhard Heckel and Helmut B{\"o}lcskei.
\newblock Subspace clustering via thresholding and spectral clustering.
\newblock In {\em 2013 IEEE International Conference on Acoustics, Speech and Signal Processing}, pages 3263--3267. IEEE, 2013.

\bibitem{pesce2022subspace}
Luca Pesce, Bruno Loureiro, Florent Krzakala, and Lenka Zdeborov{\'a}.
\newblock Subspace clustering in high-dimensions: Phase transitions \& statistical-to-computational gap.
\newblock {\em Advances in Neural Information Processing Systems}, 35:27087--27099, 2022.

\bibitem{heckel2015robust}
Reinhard Heckel and Helmut B{\"o}lcskei.
\newblock Robust subspace clustering via thresholding.
\newblock {\em IEEE transactions on information theory}, 61(11):6320--6342, 2015.

\bibitem{2020SciPy-NMeth}
Pauli Virtanen, Ralf Gommers, Travis~E. Oliphant, Matt Haberland, Tyler Reddy, David Cournapeau, Evgeni Burovski, Pearu Peterson, Warren Weckesser, Jonathan Bright, St{\'e}fan~J. {van der Walt}, Matthew Brett, Joshua Wilson, K.~Jarrod Millman, Nikolay Mayorov, Andrew R.~J. Nelson, Eric Jones, Robert Kern, Eric Larson, C~J Carey, {\.I}lhan Polat, Yu~Feng, Eric~W. Moore, Jake {VanderPlas}, Denis Laxalde, Josef Perktold, Robert Cimrman, Ian Henriksen, E.~A. Quintero, Charles~R. Harris, Anne~M. Archibald, Ant{\^o}nio~H. Ribeiro, Fabian Pedregosa, Paul {van Mulbregt}, and {SciPy 1.0 Contributors}.
\newblock {{SciPy} 1.0: Fundamental Algorithms for Scientific Computing in Python}.
\newblock {\em Nature Methods}, 17:261--272, 2020.

\end{thebibliography}

\end{document}